\documentclass[11pt]{article}  

\usepackage[square, authoryear]{natbib}
\usepackage[all]{xy}

\usepackage{graphicx}

\usepackage{caption}
\usepackage{subcaption}

\usepackage{amssymb}
\usepackage{latexsym}
\usepackage[english]{babel}
\usepackage{amsmath}
\usepackage{makeidx}
\usepackage[utf8]{inputenc}
\usepackage{verbatim}
\usepackage[pagewise]{lineno}
\usepackage{amsthm}
\usepackage{tikz-cd} 
\usepackage{stmaryrd} 

\newtheorem{theorem}{Theorem}[section]

\newtheorem{lemma}[theorem]{Lemma}
\newtheorem{proposition}[theorem]{Proposition}

\theoremstyle{definition}
\newtheorem{definition}[theorem]{Definition}
\newtheorem{remark}[theorem]{Remark}
\newtheorem{example}{Example}

\newcommand{\R}{\ensuremath{\mathbb{R}}}

\newcommand{\C}{\ensuremath{\mathcal{C}}}
\newcommand{\D}{\ensuremath{\mathcal{D}}}
\newcommand{\Es}{\ensuremath{\mathbb{S}}}

\newcommand{\U}{\ensuremath{\mathcal{U}}}
\newcommand{\M}{\ensuremath{\mathcal{M}}}
\newcommand{\J}{\ensuremath{\mathcal{J}}}

\newcommand{\F}{\ensuremath{\mathbb{F}}}

\DeclareMathOperator{\arcsinh}{arcsinh}

\begin{document}
	\title{Exact discrete Lagrangian mechanics for nonholonomic mechanics }
	\author{
		{\bf\large Alexandre Anahory Simoes}\hspace{2mm}
		\vspace{1mm}\\
		{\small  Instituto de Ciencias Matem\'aticas (CSIC-UAM-UC3M-UCM)} \\
		{\small C/Nicol\'as Cabrera 13-15, 28049 Madrid, Spain}\\
		{\it\small e-mail: \texttt{alexandre.anahory@icmat.es }}\\
		\vspace{2mm}\\
	{\bf\large Juan Carlos Marrero}\hspace{2mm}
	\vspace{1mm}\\
	{\it\small 	ULL-CSIC Geometr{\'\i}a Diferencial y Mec\'anica Geom\'etrica,}\\
	{\it\small  {Departamento de Matem\'aticas, Estad{\'\i}stica e I O, }}\\
	{\it\small  {Secci\'on de Matem\'aticas, Facultad de Ciencias}}\\
	{\it\small  Universidad de la Laguna, La Laguna, Tenerife, Canary Islands, Spain}\\
	{\it\small e-mail: \texttt{{jcmarrer@ull.edu.es}} }\\
		\vspace{2mm}\\
		{\bf\large David Martín de Diego}\hspace{2mm}
		\vspace{1mm}\\
		{\small  Instituto de Ciencias Matem\'aticas (CSIC-UAM-UC3M-UCM)} \\
	{\small C/Nicol\'as Cabrera 13-15, 28049 Madrid, Spain}\\
		{\it\small e-mail:  \texttt{david.martin@icmat.es} }\\		
	}

\date{}

\maketitle

\vspace{0.5cm}
\begin{abstract}

	We construct the exponential map associated to a nonholonomic system that allows us to define  an exact discrete nonholonomic constraint submanifold.  We reproduce the continuous  nonholonomic flow as a discrete flow on this discrete constraint submanifold  deriving an   exact  discrete version of the nonholonomic equations. Finally, we derive a  general  family of nonholonomic integrators.

\end{abstract}

\let\thefootnote\relax\footnote{\noindent AMS {\it Mathematics Subject Classification (2010)}. Primary 34C25; Secondary  37E40,
	70K40.\\
	\noindent Keywords. nonholonomic mechanics, discrete mechanics, exponential map,  geometric integration  }

\section{Introduction}

Many mechanical systems of interest in applications possess underlying geometric structures that are preserved along the time evolution as, for instance, energy and  other constant of the motions, reversibility, symplecticity... Therefore, when we  implement numerical simulations it is interesting to exactly preserve some of these geometric properties  to improve  the quantitative and qualitative accuracy and long-time stability of the proposed methods. This is precisely the main  idea behind geometric integration \cite{serna, hairer, blanes} and, in particular, of discrete mechanics and variational integrators \cite{marsden-west}. In this last case, the construction of an exact discrete Lagrangian is a crucial element for the analysis of the error between the continuous  trajectory and the numerical simulation  derived by a variational integrator (see also \cite{marsden-west, PatrickCuell} and \cite{Sato, GrFe} for forced systems).
However,  an open question is how  to derive the exact discrete version for nonholonomic mechanics (see  \cite{McLP2005} for an attempt) and this is the main topic of the present paper. The importance of this problem was point out as an open problem by R.I. MacLachlan and C. Scovel: 
\begin{quote}
	{\sl The problem for the more general class of non-holonomic constraints is still open, as is the question of the correct analogue of symplectic integration for non-holonomically constrained Lagrangian systems} \cite{MR1406809}
\end{quote}

The importance of nonholonomic systems appears since they model  mechanical systems subjected to velocity constraints which are not derivable from position or holonomic constraints and their equations are not obtained using variational techniques. This is the case, for instance, of rolling without slipping.  These systems are of  considerable interest  since the velocity or nonholonomic constraints are present in a great variety of mechanical systems in engineering and robotics (see \cite{Cortes, Bloch} and references therein). 
However, at the moment, there is no consensus in the scientific community on the best geometrical methods for numerically integrate a non-holonomic system but several possibilities were proposed inspired in the geometry of nonholonomic systems and suitable discretizations of Lagrange-d'Alembert principle \cite{modin}. 
We think that one of the reasons for these plethora of so different methods (see \cite{CM2001, McLP2005, FIM2008, BZ, FBO2012, MR3993177, IMdDM2008}, among others) can be related with the difficulty to find an exact discrete version of the nonholonomic mechanics as it happens in the case of  Lagrangian  mechanics. This is precisely the main contribution  of the paper. First, we study how to describe geometrically the exact discrete space where the nonholonomic flow evolves as a submanifold of the Cartesian product of two copies of the configuration space and then we construct an exact discrete version of nonholonomic dynamics. 
Our construction allows us to motivate a new class of nonholonomic integrators: {\sl modified Lagrange-d'Alembert integrators} (see \cite{parks} for an application of similar methods  to Dirac systems).

The outline of the paper is the following: in Section \ref{chapter2}  we review the theory of Lagrangian mechanics three-fold: unconstrained, nonholonomic constrained and forced. 
In Section 3 we construct the nonholonomic exponential map using the theory of second-order differential equations restricted to the constraint submanifold. The main result is summarized in Theorem \ref{Sodeexp}. The nonholonomic exponential map allows us to introduce an important
geometric object: the exact discrete nonholonomic constraint submanifold. In Section 4 we will review discrete Lagrangian mechanics for unconstrained systems and the discrete Lagrange-d'Alembert principle for discrete forced mechanics. In Section 5, we introduce the exact discrete flow for nonholonomic mechanics and we derive an integrator having it as a particular solution. With this motivation we construct a new family of nonholonomic integrators based on the  properties of the exact discrete equations. This theory is applied to several examples showing in numerical computations the excellent behaviour of the energy. Finally, we discuss in Section \ref{towards} new directions to find a completely intrinsic version of nonholonomic mechanics as a discrete version of the recently proposed continuous setting \cite{MR2492630, MR2660714}. 

Unless stated otherwise, all the maps and manifolds in this paper are smooth. Einstein's summation convention is used along the paper. 

\section{Continuous  Lagrangian mechanics}\label{chapter2}

\subsection{Unconstrained systems}

A \textit{mechanical system} is a pair formed by a smooth manifold $Q$ called the \textit{configuration space} and a smooth function $L:TQ\rightarrow\R$ on its tangent bundle called the \textit{Lagrangian} 
\cite{AM78, MR1021489}. If the system is not subjected to any constraint or external forces, a \textit{motion} of the mechanical system is a solution of the \textit{Euler-Lagrange equations}, whose expression on natural coordinates relative to a chart $(q^{i})$ for $Q$ and the induced coordinates $(q^i, \dot{q}^i)$ on $TQ$ is 
\begin{equation} \label{EL}
\frac{d}{dt}\left(\frac{\partial L}{\partial \dot{q}^{i}}\right) - \frac{\partial L}{\partial q^{i}}=0. 
\end{equation}

As it is well-known these equations are obtained by minimizing the action functional defined over curves with fixed end points. Denote the set of twice differentiable curves with fixed end-points $q_0, q_1 \in Q$ by $$C^2 (q_0,q_1)= \{q:[0,h]\longrightarrow Q| \ q(\cdot) \ \text{is} \ C^2, q(0)=q_0, q(h)=q_1\}.$$ Then the \textit{action functional} is defined by $$\J:C^2 (q_0,q_1)\longrightarrow \R, \ \ q(\cdot)\mapsto \J(q(\cdot))=\int_{0}^{h} L(q(t),\dot{q}(t)) \ dt.$$

We can also express these equations using the geometric ingredients  on the tangent bundle. Let $\tau_Q: TQ\rightarrow Q$ be the canonical tangent projection which in  coordinates is given by $(q^i, \dot{q}^i)\longrightarrow (q^i)$.  
The \textit{vertical lift} of a vector $v_q\in T_qQ=\tau^{-1}_Q(q)$ to $T_{u_q}TQ$, with $u_q\in T_qQ$  is defined by
\[
(v_q)^V_{u_q}=\left.\frac{d}{dt}\right|_{t=0}(u_q + t v_q)
\]
and the \textit{Liouville vector field} on $TQ$ is
$$\Delta(v_q)=\left.\frac{d}{dt}\right|_{t=0}(v_q + t v_q)=(v_q)^V_{v_q}.$$
The \textit{vertical endomorphism} $S:TTQ\rightarrow TTQ$ is defined by
$$S(X_{v_q})=(T_{v_q} \tau_M (X_{v_q}))^V_{v_q}.$$ In local coordinates, $\Delta(q^i, v^i)=v^{i}\frac{\partial}{\partial \dot{q}^{i}}$ and $S(X^{i}\frac{\partial}{\partial q^{i}}+X^{n+i}\frac{\partial}{\partial \dot{q}^{i}})=X^{i}\frac{\partial}{\partial \dot{q}^{i}}$.

Other notion that will be used later is that of the vertical lift of a vector field on $Q$ to $TQ$. Let $X\in {\mathfrak X}(Q)$, the \textit{vertical lift} of $X$ is the vector field on $TQ$ defined by:
\[
X^V(v_{q})=\left.\frac{d}{dt}\right|_{t=0}(v_{q}+tX(q)))=(X(q))^V_{v_{q}},\; \ \forall v_{q}\in T_{q} Q.
\]
Locally,
\begin{equation}\label{vertical:lift}
	X^V=X^i\frac{\partial}{\partial\dot{q}^i}
\end{equation}
where $X=X^i\frac{\partial}{\partial q^i}$.

Denote by $\{\Phi^X_t\}$ the flow of a vector field $X\in {\mathfrak X}(Q)$.  We can also define the \textit{complete lift}  $X^C\in {\mathfrak X}(TQ)$ of $X$   in terms of its flow. We say that $X^C$ is the vector field on $TQ$ with flow $\{T\Phi^X_t\}$. In other words,
\[
X^C(v_q)=\left.\frac{d}{dt}\right|_{t=0}\left(T_q\Phi^X_t(v_q)\right)\; .
\]
In coordinates
\begin{equation}\label{complete:lift}
	X^C=X^i\frac{\partial}{\partial q^i}+\dot{q}^j\frac{\partial X^i}{\partial q^j}\frac{\partial}{\partial \dot{q}^i} \ .
\end{equation}
Note that, if $q^{i}(t)$ are the local coordinates of a curve on $Q$, then using \eqref{vertical:lift} and \eqref{complete:lift}, it is easy to prove that such a curve is a solution of Euler-Lagrange equations \eqref{EL} if and only if
\begin{equation*}
	X^{C}(L) (q,\dot{q})-\frac{d}{dt}\left(X^{V}(L)(q,\dot{q})\right)=0, \quad \forall \ X\in \mathfrak{X}(Q).
\end{equation*}

When the function $L$ is \textit{regular} that is, the matrix $\text{Hess}(L):=\left(\frac{\partial^2 L}{\partial \dot{q}^i \partial \dot{q}^j}\right)$ is non-singular, equations \eqref{EL} may be written as a system of second-order differential equations obtained by computing the integral curves of the unique vector field $\Gamma_L$ satisfying
\begin{equation}
i_{\Gamma_L}\omega_L=dE_L, \label{Lvf}
\end{equation}
where $\omega_L =-d(S^* dL)$ and $E_L =\Delta L-L$ are the \textit{Poincar\'e-Cartan 2-form} and the \textit{energy function}, respectively. Moreover,  $\Gamma_L$ verifies that $S(\Gamma_L)=\Delta$, that is, $\Gamma_L$ is a SODE vector field on $Q$ (see \cite{MR1021489}). Observe that regularity of $L$ is equivalent to $\omega_L$ being symplectic and therefore to the uniqueness of solution for equation \eqref{Lvf}. In effect, the local expression of the  Poincar\'e-Cartan 2-form  is
$$\omega_L=\frac{\partial^2 L}{\partial \dot{q}^i \partial q^j} d q^{i}\wedge d q^{j}+\frac{\partial^2 L}{\partial \dot{q}^i \partial \dot{q}^j} dq^{i}\wedge d\dot{q}^{j}.$$

Now, we move on to a brief description of standard  Hamiltonian mechanics. The cotangent bundle $T^*Q$ of a differentiable manifold $Q$ is equipped with a canonical exact symplectic structure $\omega_Q=-d\theta_Q$, where $\theta_Q$ is the canonical 1-form on $T^*Q$ defined  by
\[
(\theta_Q)_{\alpha_q}(X_{\alpha_q})=\langle \alpha_q, T_{\alpha_q}\pi_Q(X_{\alpha_q})\rangle
\]
where $X_{\alpha_q}\in T_{\alpha_q}T^*Q$, $\alpha_q\in T_q^*Q$ and 
$\pi_Q: T^*Q\rightarrow Q$ is the canonical projection which in canonical coordinates   is $(q^i, p_i)\rightarrow (q^i)$. 
In canonical bundle coordinates these become
\[
\theta_Q= p_i\, dq^i\; ,\ 
\omega_Q= dq^i\wedge dp_i\; .
\]
Given a Hamiltonian function $H: T^*Q\rightarrow {\mathbb R}$ we define the Hamiltonian vector field $X_H$ by
\[
\imath_{X_H}\omega_Q=dH\; 
\]
The integral curves of $X_H$ are determined by Hamilton's equations: 
\[
	\frac{dq^i}{dt}=\frac{\partial H}{\partial p_i}\; ,\qquad 
	\frac{dp_i}{dt}=-\frac{\partial H}{\partial q^i}\; .
\]
We can define the Legendre transformation $\mathbb{F} L: TQ\rightarrow T^*Q$ by:
\[
\langle \mathbb{F} L (u_q), v_q\rangle=\left.\frac{d}{dt}\right|_{t=0}L(u_q + t v_q)
\]
and if $L$ is regular, its Legendre transformation is a local diffeomorphism. In local coordinates $\mathbb{F}L (q^i, \dot{q}^i)=(q^i, \frac{\partial L}{\partial \dot q^i})$. Defining 
$H = E_L \circ \left(\mathbb{F}L\right)^{-1}$
we have that the solutions of $\Gamma_L$ and $X_H$ are $\mathbb{F}L$-related. 
An extensive account of this subject is contained in \cite{AM78, MR1021489}, for instance.

\subsection{Forced mechanics}\label{forced:mechanics:section}

Now, we also add into the picture external forces. An \textit{external force} can be interpreted as a fiber-preserving map denoted by $F:TQ\rightarrow T^{*}Q$ satisfying $\pi_Q \circ F = \tau_Q$.  In canonical bundle coordinates $(q^i, p_i)$ on $T^*Q$ we have that $\pi_Q(q^i, p_i)=(q^i)$, thus $F(q^i, \dot{q}^i)=(q^i, F_i(q^{i}, \dot{q}^{i}))$.

\[
\xymatrix{
	TQ \ar[rr]^{F}\ar[rd]_{\tau_Q}&& T^*Q\ar[ld]^{\pi_Q}\\\
	&Q&
}
\]
 It is well-know that to each such map we can associate a semibasic one-form on $TQ$ defined by $$\langle \mu_{F}(v_q), W \rangle=\langle F(v_q),T\tau_{Q}(W) \rangle, \ \ v_q\in TQ \ \text{and} \ W\in T_{v_q}TQ.$$ In coordinates
 $
 \mu_F=F_i(q^i, \dot{q}^i)\, dq^i\; .
 $

A system described by a Lagrangian function $L:TQ\rightarrow\R$ and subjected to an external force $F$, satisfies the \textit{Lagrange-d'Alembert principle}, which asserts that a motion of this system between two fixed points $q_0,q_1\in Q$ is a curve $q\in C^{2}(q_0,q_1)$ satisfying
\begin{equation}
\left.\frac{d}{ds}\right|_{s=0} \int_{0}^{h} L(q(t,s),\dot{q}(t,s)) \ dt + \int_{0}^{h} \left\langle F(q(t),\dot{q}(t)), \frac{\partial q}{\partial s}(t,0)\right\rangle \ dt=0,
\end{equation}
for all smooth variations $q(s)\in C^{2}(q_0,q_1)$ of $q$. This is locally equivalent to the \textit{forced Euler-Lagrange} equations
\begin{equation} \label{FEL}
\frac{d}{dt}\left(\frac{\partial L}{\partial \dot{q}^{i}}\right) - \frac{\partial L}{\partial q^{i}}=F_{i}\; .
\end{equation}

As in the case of unconstrained systems, it is easy to see using \eqref{FEL}, that a curve $q(t)$ on $Q$ satisfies the forced Euler-Lagrange equations if and only if
\begin{equation*}
		X^{C}(L) (q,\dot{q})-\frac{d}{dt}\left(X^{V}(L)(q,\dot{q})\right)=\langle F(q,\dot{q}),X\circ q \rangle, \quad \forall \ X\in \mathfrak{X}(Q).
\end{equation*}

If $L$ is regular, then the solutions of equations \eqref{FEL} are integral curves of a SODE vector field on $Q$ denoted by $\Gamma_{(L,F)}$, called \textit{forced Lagrangian vector field} which is the unique vector field satisfying
\begin{equation}\label{fLvf}
i_{\Gamma_{(L,F)}}\omega_{L}=dE_{L}-\mu_{F}.
\end{equation}

Now, we move onto the Hamiltonian description of systems subjected to external forces. Given a Hamiltonian function  $H: T^*Q\rightarrow {\mathbb R}$ we may construct the transformation $\mathbb{F}H: T^*Q\rightarrow TQ$ where $\langle \beta_q, \mathbb{F}H(\alpha_q)\rangle={\frac{d}{dt}\big |_{ t=0}H(\alpha_q+t\beta_q)}$.
In coordinates, $\mathbb{F}H(q^i, p_i)=(q^i, \frac{\partial H}{\partial p_i}(q,p))$. 
We say that the Hamiltonian is regular if $\mathbb{F}H$ is a local diffeomorphism, which in local coordinates is equivalent to the regularity of the Hessian matrix whose entries are:
\begin{equation*}
M^{i j} = \frac{\partial^2 H}{\partial p_i\partial p_j}.
\end{equation*}
Consider now the external force previously defined in the Lagrangian description and denote $F^H = F\circ \mathbb{F}H: T^*Q\rightarrow T^*Q$. 
\[
\xymatrix{
	T^*Q \ar[rr]^{F^H}\ar[rd]_{\pi_Q}&& T^*Q\ar[ld]^{\pi_Q}\\\
	&Q&
}
\]
It is possible to modify the Hamiltonian vector field $X_H$ to obtain the forced Hamilton's equations as the integral curves of the vector field $X_H+Y^v_F$ where the vector field
$Y^v_F\in {\mathfrak X}(T^*Q)$ is defined by
\[
Y^v_F(\alpha_q)=\frac{d}{dt}\Big|_{t=0}(\alpha_q+t F^H(\alpha_q))\; .
\]
We will say the the forced Hamiltonian system is determined by the pair $(H, F^H)$.

Locally, 
\[
Y^v_F=F_i\left(q^j, \frac{\partial H}{\partial p_j}(q, p)\right)\frac{\partial}{\partial p_i}=F^H_i(q, p)\frac{\partial}{\partial p_i}
\]
modifying Hamilton's equations as follows:
\begin{eqnarray}\label{qwe}
\frac{dq^i}{dt}&=&\frac{\partial H}{\partial p_i}(q,p)\; ,\\
\frac{dp_i}{dt}&=&-\frac{\partial H}{\partial q^i}(q,p)+F^H_i(q,p)\; .
\end{eqnarray}

\subsection{Nonholonomic systems}\label{nho}

A \textit{nonholonomic system} is defined by the triple $(Q, L, \D)$  where $L: TQ\rightarrow {\mathbb R}$ is a Lagrangian function and $\D$ is a nonintegrable distribution  on the configuration manifold $Q$. The distribution $\D$ restricts the velocity vectors of motions to lie on $\D$ without imposing any restriction on the configuration space. Note that if the distribution was integrable, then the manifold $Q$ would be foliated by immersed submanifolds of $Q$ whose tangent space at each point coincides with the subspace given by the distribution at that point. Hence, motions of these systems would be confined to submanifolds $N\subseteq Q$ (the leaves of the foliation). In this way, we can consider this case as a holonomic system specified by $(N,L|_{N})$. This class of constraints is called \textit{holonomic constraints}. See  \cite{Bloch} for more details.

Locally, the nonholonomic constraints are given by a set of $k$ equations that are linear on the velocities $$\mu^{a}_{i}(q)\dot{q}^{i}=0,$$ where $1\leqslant a\leqslant k$ and the rank of $\D$ is $\text{dim}(Q)-k$. From other point of view, these equations define the vector subbundle $\D^{o}\subseteq T^{*}Q$, called the \textit{annihilator} of $\D$, spanned at each point by the one forms $\{\mu^{a}\}$ locally given by $\mu^{a}=\mu^{a}_{i}(q) dq^{i}$. Observe that with this relationship, we can identify the distribution $\D$ with a submanifold of the tangent bundle that we also denote by $\D$.  

In nonholonomic mechanics, the  equations of motion are completely determined by the   \textit{La\-gran\-ge-d'Alembert principle}. This principle states that a curve $q(\cdot)\in C^2(q_0, q_1)$ is an admissible motion of the system if
\[
\delta\mathcal{J}=\delta\int^{h}_{0}L( q(t), \dot q(t))\,  dt=0\, ,
\]
for all variations such that $\delta q (t)\in\mathcal{D}_{q(t)}$, $0\leq t\leq h$, with $\delta q(0)=\delta q(h)=0$. The velocity of the curve itself must also satisfy the constraints $\dot{q}(t)\in \D_{q(t)}$.   
From the Lagrange-d'Alembert principle, we arrive at the well-known \textit{nonholonomic equations}
\begin{align} \label{LdA}
& \frac{d}{dt}\left(\frac{\partial L}{\partial \dot{q}^{i}}\right) - \frac{\partial L}{\partial q^{i}}=\lambda_{a}\mu^{a}_{i}(q)\\
& \mu^{a}_{i}(q)\dot{q}^{i}=0, \label{nhconstr}
\end{align}
for some Lagrange multipliers $\lambda_{a}$, which may be determined with the help of the constraint equations.

In more geometric terms,  equations \eqref{LdA} and \eqref{nhconstr} are the differential equations for a SODE $\Gamma_{nh}$ on $\D$ satisfying the equations
\begin{eqnarray}\label{nhequation}
i_{\Gamma_{nh}}\omega_L-dE_L\in \Gamma(F^o), \\
\Gamma_{nh} \in \mathfrak{X}(\D), \label{nhtagency}
\end{eqnarray}
where $F^{o}=S^*((T\D)^{o})$ is the annihilator of a distribution $F$ on $TQ$ defined along $\D$ and $\Gamma(F^o)$ is the space of sections of $F^{o}$. The nonholonomic system is said to be \textit{regular} if the following conditions are satisfied (see \cite{LMdD1996}:
\begin{enumerate}
\item $\text{dim}(T_{v}\D)^{o}=\text{dim}F_{v}^{o}$ (\textit{admissibility condition});
\item $T_{v}\D\cap (\sharp)_{v}(F_{v}^{o})=\{0\}$ for all $v\in\D$ (\textit{compatibility condition}).
\end{enumerate}
The \textit{sharp isomorphism} $\sharp:T^{*}(TQ)\rightarrow T(TQ)$ is the inverse map to the \textit{flat isomorphism} defined by $\flat(X)=i_{X}\omega_{L}$. If the nonholonomic system is regular, then equations \eqref{nhequation} and \eqref{nhtagency} have a unique solution denoted by $\Gamma_{nh}$ whose integral curves satisfy equations \eqref{LdA} and \eqref{nhconstr}.

To each of the one-forms $\mu^{a}\in \D^o$ we associate the constraint functions  $\Phi^{a}:TQ\rightarrow \R$ defined by $\Phi^{a}(v_q)=\langle \mu^{a}(q),v_q \rangle$ or $\Phi^{a}(q, \dot{q})= \mu^{a}_i(q)\dot{q}^i $.  
In local coordinates, equation \eqref{nhequation} may be written like
$$i_{\Gamma_{nh}}\omega_L-dE_L=\lambda_{a} S^{*}(d\Phi^{a})=\lambda_{a} \mu^{a}_{i} dq^{i},$$
for some Lagrange multipliers $\lambda_{a}$. Therefore, a solution $\Gamma_{nh}$ of \eqref{nhequation} is of the form $\Gamma_{nh}=\Gamma_{L}+\lambda_{a} Z^{a}$, where $Z^{a}=\sharp(\mu^{a}_{i} dq^{i})$. The Lagrange multipliers may be computed by imposing the tangency condition \eqref{nhtagency}, which is equivalent to
\begin{equation}\label{lambda}
0=\Gamma_{nh}(\Phi^{b})=\Gamma_{L}(\Phi^{b})+\lambda_{a} Z^{a}(\Phi^{b}), \ \ \text{for} \ b=1,...,n-k.
\end{equation}
This equation has a unique solution for the Lagrange multipliers if and only if the matrix $C=(\mathcal{C}^{a b})=(Z^{a}(\Phi^{b}))$ is invertible at all points of $\D$, which is equivalent to the compatibility condition (cf. \cite{LMdD1996}).

Recall from symplectic geometry that $F^{\bot}=\sharp(F^{o})$ for any distribution $F$, where $\bot$ denotes the symplectic complement relative to $\omega_{L}$. Hence, the admissibility and compatibility conditions also imply the following Whitney sum decomposition $$T(TQ)|_{\D}=T\D\oplus F^{\bot},$$
to which we may associate two complementary projectors $P:T(TQ)|_{\D}\rightarrow T\D$ and $P':T(TQ)|_{\D}\rightarrow F^{\bot}$ with coordinate expressions $$P(X)=X-\mathcal{C}_{a b}\, d\Phi^{b}(X)Z^{a}, \quad  P'(X)=\mathcal{C}_{a b}\, d\Phi^{b}(X)Z^{a},$$
where $C_{ab}$ are the entries of the inverse matrix  $C^{-1}$ of $C$.

\begin{proposition}
	The nonholonomic dynamics is given by $$\Gamma_{nh}=P(\Gamma_{L}|_{\D}).$$
\end{proposition}

Indeed, under all the assumptions we have considered so far, we can compute the Lagrange multipliers to be \begin{equation}\label{lambda-def}
\lambda_{a}=-\mathcal{C}_{a b}\Gamma_{L}(\Phi^{b}),
\end{equation} from where the result follows. So, under the admissibility and compatibility conditions, the nonholonomic system $(L,\D)$ is said to be regular. For more details see \cite{Bloch} or \cite{LMdD1996}.

\begin{remark}\label{rem1}
	Note that, under the admissibility and compatibility conditions, nonholonomic mechanics can be interpreted as ``restricted forced systems", in the sense that we can define the nonholonomic external force $F_{nh}:\D\rightarrow T^{*}Q$ which makes \eqref{LdA} forced Euler-Lagrange equations. In coordinates, $F_{nh}(v_q)=\lambda_{a}(v_q)\mu_{i}^{a}(q)dq^{i}$
	where the $\lambda_a$ are given in expression (\ref{lambda-def}). Moreover, as in the case of forced Lagrangian systems, if $q(t)$ is a curve on $Q$ such that $\dot{q}(t)\in \D$, then such a curve is a solution of the nonholonomic equations \eqref{LdA} if and only if
	\begin{equation}\label{LdA:vector:field}
		X^{C}(L) (q,\dot{q})-\frac{d}{dt}\left(X^{V}(L)(q,\dot{q})\right)=\langle F_{nh}(q,\dot{q}),X\circ q \rangle, \quad \forall \ X\in \mathfrak{X}(Q).
	\end{equation}
\end{remark}

Taking the restriction of the Lagrangian $L: TQ\rightarrow \R$ to $\D$ denoted by $l: \D\rightarrow \R$ we can construct the \textit{nonholonomic Legendre map} 
\[
\mathbb{F} l: \D\longrightarrow \D^*\; ,
\]
as 
\[
\langle \mathbb{F} l (u_q), v_q\rangle=\left.\frac{d}{dt}\right|_{t=0}l(u_q + t v_q)
\]
for $u_q, v_q\in \D$.
Under the admissibility and compatibility assumptions, the map $\mathbb{F} l$ is a local diffeomorphism and we can transport  the vector field $\Gamma_{nh}\in \mathfrak{X}(\D)$ to a vector field $\bar\Gamma_{nh}\in \mathfrak{X}(\D^*)$ which represents the almost-Hamiltonian dynamics on $\D^*$ \cite{MR2492630,MR2660714}.

\begin{example}\label{Nhparticle}
	We will introduce here an example of a simple nonholonomic system to which we will get back all along the text: the \textit{nonholonomic particle}. Consider a mechanical system in the configuration manifold $Q=\R^3$ defined by the Lagrangian
	$$L(x,y,z,\dot{x},\dot{y},\dot{z})=\frac{1}{2}(\dot{x}^2 +\dot{y}^2 +\dot{z}^2)$$
	and subjected to the nonholonomic constraint $\dot{z}-y \dot{x}=0$. The one-form $\mu=dz-y \ dx$ spans the vector subbundle $\D^{o}$, which is the annihilator of the distribution $$\D=\text{span}\left\{\frac{\partial}{\partial x}+y \frac{\partial}{\partial z}, \frac{\partial}{\partial y} \right\}.$$
	Then the equations of motion of this system are given by Lagrange-d'Alembert equations \eqref{LdA} and \eqref{nhconstr}, which in this case hold
	\begin{equation}
		\begin{cases}
			\ddot{x}=- \lambda y \\
			\ddot{y}=0 \\
			\ddot{z}=\lambda \\
			\dot{z}-y\dot{x}=0
		\end{cases}
		\quad
		\Rightarrow
		\quad
		\begin{cases}
			\ddot{x}=-y\frac{\dot{x}\dot{y}}{1+y^2} \\
			\ddot{y}=0 \\
			\ddot{z}=\frac{\dot{x}\dot{y}}{1+y^2} \\
			\dot{z}-y\dot{x}=0,
		\end{cases}
	\end{equation}
	where the value of $\lambda$ is computed with the help of the constraints. These equations have an explicit solution given by
	\begin{equation}\label{nhsolution}
		\begin{cases}
			x_{nh}(t)= \frac{\dot{x}_0}{\dot{y}_0}\sqrt{y_0^2+1}(\arcsinh(\dot{y}_0 t+y_0)-\arcsinh(y_0))+x_0\\
			y_{nh}(t)=\dot{y}_0 t+y_0 \\
			z_{nh}(t)= \frac{\dot{x}_0}{\dot{y}_0}\sqrt{y_0^2+1}(\sqrt{(\dot{y}_0 t+y_0)^2+1}-\sqrt{y_0^2+1})+z_0,
			\quad
			\text{if} \ \ \dot{y}_0\neq 0,
		\end{cases}
	\end{equation}
	or
	\begin{equation}\label{nhsolution0}
	\begin{cases}
	x_{nh}(t)= \dot{x}_0 t+x_0\\
	y_{nh}(t)=y_0 \\
	z_{nh}(t)= y_0 \dot{x}_0 t+z_0,
	\quad
	\text{if} \ \ \dot{y}_0= 0.
	\end{cases}
	\end{equation}
	
\end{example}


\section{The nonholonomic exponential map} \label{sec:exponential-map}

In this section, we will define  the \textit{nonholonomic exponential map} using an arbitrary SODE extension $\Gamma \in {\mathfrak X}(TQ)$ of $\Gamma_{nh}$, that is, 
\[
\Gamma_{nh}={\Gamma}|_{\D}
\]
and the standard definition of exponential map $\text{exp}_h^{\Gamma}$ for a SODE $\Gamma$ on $TQ$ (see \cite{MMdDM2016} and references therein).  
 
In fact, if $q_{0}\in Q$ then we may consider the \textit{exponential map at $q_{0}$}, which is defined as follows:
\[
\text{exp}_{h, q_0}^{\Gamma}(v_{q_{0}})=\tau_Q(\phi_h^\Gamma (v_{q_{0}})), \quad v_{q_{0}}\in T_{q_{0}}Q
\]
where $\{\phi_h^\Gamma \}$ is the flow of $\Gamma$ for a sufficiently small non-negative number $h\geq 0$. Therefore  for $v_{q_{0}}\in D_{q_{0}}$: 
\[
\text{exp}_{h, q_0}^{\Gamma}(v_{q_{0}})=\text{exp}_{h, q_0}^{\Gamma_{nh}}(v_{q_{0}})
=\tau_Q(\phi_h^{\Gamma _{nh}}(v_{q_{0}}))
\]
 which  does not depend on the particular extension $\Gamma$. Here $\{\phi_h^{\Gamma _{nh}}\}$ denotes the flow of $\Gamma_{nh}$ evaluated at time $h$. 

Denote also by
\begin{equation}\label{exp:gamma}
	\text{exp}_{h}^{\Gamma}(v_q)=(\tau_Q(v_q), \text{exp}_{h, \tau_Q(v_q)}^{\Gamma}(v_q))\subseteq Q\times Q, \quad q\in Q, v_{q}\in T_{q} Q.
\end{equation}

The following theorem gives a precise statement of the previous discussion (see also \cite{MMdDM2016}).
\begin{theorem}\label{Sodeexp}
Given a SODE $\Gamma\in\mathfrak{X}(TQ)$ on a manifold $Q$, $q_{0}\in Q$ and $h$ a sufficiently small positive number, its exponential map defined on the open subset of $T_{q_{0}}Q$
\begin{equation*}
	M_{h,q_{0}}^{\Gamma}=\{v\in T_{q_{0}}Q\;  |\;  \phi_{t}^{\Gamma}(v) \  \text{is defined for} \ t\in [0,h]\},
\end{equation*}
and denoted by $\text{exp}_{h,q_{0}}^{\Gamma}: M_{h,q_{0}}^{\Gamma} 	\rightarrow Q$ is a diffeomorphism from an open subset $V\subseteq M_{h,q_{0}}^{\Gamma}$ to an open subset $U\subseteq Q$ with $q_{0}\in U$, i.e.,
\begin{equation*}
	T_{v_0}\text{exp}_{h,q_{0}}^{\Gamma}:T_{v_{0}}M_{h,q_{0}}^{\Gamma}:\rightarrow T_{\text{exp}_{h,q_{0}}^{\Gamma}(v_{0})}Q
\end{equation*}
is non-singular 	for all  $v_{0}\in V$.
\end{theorem}

\begin{proof}
If $\Gamma$ is a spray, that is,  its components are homogeneous functions of degree 2, then  we fall in the usual proof that the exponential map is a diffeomorphism, which we can find in the literature from Riemannian geometry \cite{docarmo}. However in the general case we can not use the standard argument.

Note that, using that $M_{h,q_{0}}^{\Gamma}$ is an open subset of $T_{q_{0}}Q$, it follows that
\begin{equation*}
	T_{v_{0}}M_{h,q_{0}}^{\Gamma}=\{ X^{V}(v_{0})\in T_{v_{0}}(TQ) \ | \ X\in\mathfrak{X}(Q) \},
\end{equation*}
where $X^{V}$ is the vertical lift to $TQ$ of $X$. Moreover, it can be easily seen that $[\Gamma,X^{V}]$ projects to $-X$, that is $$(T_v\tau_{Q})[\Gamma,X^{V}](v)=-X(\tau_{Q}(v)).$$ Recalling the identity
\begin{equation}\label{identity1}
[\Gamma,X^{V}](v)=\left.\frac{d}{dt}\right|_{t=0} T_{\phi_{t}^{\Gamma}(v)}\phi_{-t}^{\Gamma}(X^{V}(\phi_{t}^{\Gamma}(v))),
\end{equation}
by smoothness we have that for small $|t|>0$,
\begin{equation}\label{identity2}
T_{\phi_{t}^{\Gamma}(v)}\phi_{-t}^{\Gamma}(X^{V}(\phi_{t}^{\Gamma}(v)))=X^{V}(v)+t[\Gamma,X^{V}](v)
+{\mathcal O}(t^2)\in T_{v}TQ,
\end{equation}
which projecting implies
\begin{equation}\label{identity3}
(T_v\tau_{Q})\left(T_{\phi_{t}^{\Gamma}(v)}\phi_{-t}^{\Gamma}(X^{V}(\phi_{t}^{\Gamma}(v)))\right)=-t(X(\tau_{Q}(v)))+{\mathcal O}(t^2)\in T_{\tau_{Q}(v)}Q.
\end{equation}

Fix a small $h>0$ and the vector $v_{h, q_{0}}\in M_{h,q_{0}}$ such that $\text{exp}_{h,q_{0}}^{\Gamma}(v_{h, q_{0}})=q_0$ (see Theorem 3.1. in \cite{MMdDM2016}). Also, fix some smooth vector field $X\in\mathfrak{X}(Q)$.	Then 
\begin{equation}\label{identity4}
T_{v_{h, q_{0}}} \text{exp}^{\Gamma}_{h,q_{0}}(X^{V}(v_{h, q_{0}}))=(T_{\phi^\Gamma_{h}(v_{h, q_{0}})}\tau_{Q})\left( T_{v_{h, q_{0}}}\phi_{h}^{\Gamma}(X^{V}(v_{h, q_{0}}))\right).
\end{equation}
Choosing $t=-h$ and $v=\phi^\Gamma_{h}(v_{h, q_{0}})$ and substituting on \eqref{identity3}, we see that the left hand side exactly matches \eqref{identity4}. Therefore obtaining
\begin{equation}
T_{v_{h, q_{0}}} \text{exp}^{\Gamma}_{h,q_{0}}(X^{V}(v_{h, q_{0}}))=h X(q_0) +{\mathcal O}(h^2).
\end{equation}
In the limit when $h$ approaches zero:
\begin{equation}\label{limite}
\lim_{h\rightarrow 0} \frac{T_{v_{h, q_{0}}} \text{exp}^{\Gamma}_{h,q_{0}}(X^{V}(v_{h, q_{0}}))}{h}=X(q_0),
\end{equation}
from where we conclude that  the left hand side of \eqref{limite} gets arbitrarily close to $ X(q_0)$ when $h$ decreases to zero and we may conclude that it is an invertible map for sufficiently small values of $h>0$. Hence, we may easily see that $T_{v_{h, q_{0}}} \text{exp}^{\Gamma}_{h,q_{0}}$ is an isomorphism since if we denote  by $A(h)$ the matrix representing $T_{v_{h, q_{0}}} \text{exp}^{\Gamma}_{h,q_{0}}$ in any system of coordinates, then $\det\left(\frac{1}{h} A(h)\right)\neq 0$ from (\ref{limite}) and $\det(A(h))=h^{n}\det(\frac{1}{h} A(h))$ as we want to proof.
\end{proof}

Let $L:TQ\rightarrow \R$ be a regular Lagrangian function and $\D$ a regular distribution on $Q$ such that the non-holonomic system $(L,\D)$ is also regular and let $\Gamma_{nh}$ be the SODE on $\D$ which is solution of the non-holonomic dynamics. Denote by $\phi_t^{\Gamma_{nh}}: D \rightarrow D$ the flow of $\Gamma_{nh}$ and for $h$ a sufficiently small positive number, we consider the open subset of $\D$ given by
\begin{equation*}
	M_{h}^{\Gamma_{nh}}=\{ v\in\D \ | \ \phi_{t}^{\Gamma_{nh}}(v) \ \text{is defined for} \ t\in [0,h] \}.
\end{equation*}
Note that, if $\Gamma\in\mathfrak{X}(TQ)$ is a SODE extension of $\Gamma_{nh}$ then
\begin{equation*}
M_{h}^{\Gamma_{nh}}=\left( \bigcup_{q_{0}\in Q} M_{h,q_{0}}^{\Gamma} \right) \cap \D.
\end{equation*}

\begin{definition}
The map
\begin{align*}
\text{exp}_h^{\Gamma_{nh}}:M_{h}^{\Gamma_{nh}}\subseteq \D & \rightarrow Q\times Q \\
v_{q_0} & \mapsto (q_0,\tau_{Q}\circ\phi_h^{\Gamma_{nh}}(v_{q_0}))
\end{align*}
is called the \textit{nonholonomic exponential map} of $\Gamma_{nh}$.
\end{definition}

Note that if $\Gamma$ is a SODE on $TQ$ such that $\Gamma|_{\D}=\Gamma_{nh}$ then
\begin{equation*}
	\text{exp}_h^{\Gamma_{nh}}=\left. \left( \text{exp}_h^{\Gamma} \right)\right|_{M_{h}^{\Gamma_{nh}}},
\end{equation*}
where $\text{exp}_h^{\Gamma}:M_h^{\Gamma}=\bigcup_{q_{0}\in Q}M_{h,q_{0}}^{\Gamma}\rightarrow Q\times Q$ is the exponential map at time $h$ associated with $\Gamma$ (as it is defined in \eqref{exp:gamma}).

\begin{lemma}\label{expdiff}
If $\Gamma$ is a SODE on TQ such that $\Gamma|_{\D}=\Gamma_{nh}$, then $\text{exp}_h^{\Gamma_{nh}}=\left. \left( \text{exp}_h^{\Gamma} \right)\right|_{M_{h}^{\Gamma_{nh}}}$. Moreover, there exists an open subset $\U_{h}$ of $\D$ such that the nonholonomic exponential map $\text{exp}_h^{\Gamma_{nh}}:\U_h\rightarrow Q\times Q$ is a smooth local embedding.
\end{lemma}

\begin{proof}
As we know (see Theorem 3.4 in \cite{MMdDM2016}), the map $\text{exp}_h^{{\Gamma}}$ is a diffeomorphism from an open subset $W_{h}\subseteq M_{h}^{\Gamma}\subseteq TQ$ to an open subset $W\subseteq Q\times Q$.

Thus,
\begin{equation*}
	\left. (\text{exp}_{h}^{\Gamma}) \right|_{W_{h}\cap M_{h}^{\Gamma_{nh}}}:W_{h}\cap M_{h}^{\Gamma_{nh}}\longrightarrow W \subseteq Q\times Q
\end{equation*}	
is a smooth immersion into $W$. Since by the local embedding theorem, every smooth immersion is a local embedding, the last claim is proved if we take $\U_{h}=W_{h}\cap M_{h}^{\Gamma_{nh}}$.
\end{proof}

\begin{definition}
Define the {\sl exact discrete nonholonomic constraint submanifold} as the submanifold of $Q\times Q$ given by 
$$\M_h^{e,nh}= \text{exp}_h^{\Gamma_{nh}} (\U_h)\; .$$
\end{definition}
 
In view of Lemma \ref{expdiff}, the map $\text{exp}_h^{\Gamma_{nh}}:\U_h\rightarrow \M_{h}^{e,nh}$ is a diffeomorphism and we can define its inverse diffeomorphism, called the \textit{nonholonomic exact retraction} map $$R_{h,nh}^{e^-}:\M_{h}^{e,nh}\longrightarrow \U_h.$$

The following are commutative diagrams:

\[
\begin{tikzcd}[row sep=2.5em]
\U_h \arrow[dr,"\tau_{\D}"'] \arrow{rr}{\text{exp}_h^{\Gamma_{nh}}} && \M_{h}^{e,nh} \arrow{dl}{pr_1} \\
& Q
\end{tikzcd}
\quad
\begin{tikzcd}[row sep=2.5em]
\M_{h}^{e,nh} \arrow[dr,"pr_1"'] \arrow{rr}{R^{e^-}_{h,nh}} && \U_h \arrow{dl}{\tau_{\D}} \\
& Q
\end{tikzcd}
\]

We will also use  the map: 
$R_{h,nh}^{e^+}:\M_{h}^{e,nh}\longrightarrow \phi_{h}^{\Gamma_{nh}}(\U_h)$ defined by
\[
R_{h,nh}^{e^+}=\phi_{h}^{\Gamma_{nh}}\circ R_{h,nh}^{e^-}
\]
Note that the following diagram
\[
\begin{tikzcd}[row sep=2.5em]
	\M_{h}^{e,nh} \arrow[dr,"pr_2"'] \arrow{rr}{R^{e^+}_{h,nh}} && \phi_{h}^{\Gamma_{nh}}(\U_h) \arrow{dl}{\tau_{\D}} \\
	& Q
\end{tikzcd}
\]
is commutative.
\begin{example}
Let us get back to Example \ref{Nhparticle}, the nonholonomic particle and identify the different geometric objects involved. The nonholonomic  vector field is given by
$$\Gamma_{nh}=\dot{x}\frac{\partial}{\partial x}+\dot{y}\frac{\partial}{\partial y}+y\dot{x}\frac{\partial}{\partial z}-y\frac{\dot{x}\dot{y}}{1+y^2}\frac{\partial}{\partial \dot{x}}+\frac{\dot{x}\dot{y}}{1+y^2}\frac{\partial}{\partial \dot{z}}.$$
From \eqref{nhsolution} and \eqref{nhsolution0}, we construct its corresponding flow and nonholonomic exponential map
$$\phi_{t}^{\Gamma_{nh}}(x_0,y_0,z_0,\dot{x}_0,\dot{y}_0,y_0\dot{x}_0)=(x_{nh},y_{nh},z_{nh},\dot{x}_{nh},\dot{y}_{nh},\dot{z}_{nh}),$$
$$\text{exp}_h^{\Gamma_{nh}}(x_0,y_0,z_0,\dot{x}_0,\dot{y}_0,y_0\dot{x}_0)=(x_0,y_0,z_0,x_{nh}(h),y_{nh}(h),z_{nh}(h)).$$
We see that this is an invertible map, when we restrict the co-domain to its image, and we may explicitly compute the inverse to be
\begin{align*}
R_{h,nh}^{e^-}(x_0,y_0,z_0,x_1,y_1,z_1)= & \left(x_0,y_0,z_0,\frac{(x_1-x_0)(y_1-y_0)}{h\sqrt{y_0^2+1}(\arcsinh(y_1)-\arcsinh(y_0))},\right. \\
& \left.\frac{y_1-y_0}{h},\frac{y_0(x_1-x_0)(y_1-y_0)}{h\sqrt{y_0^2+1}(\arcsinh(y_1)-\arcsinh(y_0))}\right),
\end{align*}
in the case where $y_1\neq y_0$. Note that the domain of the map $R_{h,nh}^{e^-}$ is not $\R^3\times\R^3$, it is restricted to $\M_{h}^{e,nh}$, which explicitly means that
\begin{equation}\label{defdiscrete}
\frac{z_1-z_0}{h}-\frac{(x_1-x_0)\left(\sqrt{y_1^2+1}-\sqrt{y_0^2+1}\right)}{h(\arcsinh(y_1)-\arcsinh(y_0))}=0.
\end{equation}
In fact, let the left-hand side of equation \eqref{defdiscrete} be denoted by $\mu_d:Q\times Q\rightarrow\R$. It is a constraint function whose annihilation gives the discrete space $\M_{h}^{e,nh}$.
\end{example}

\section{Lagrangian discrete mechanics and the exact discrete Lagrangian}

\subsection{Unconstrained discrete mechanics}

We will now describe a theory of discrete mechanics on the discretized velocity space $Q\times Q$ \cite{marsden-west}. Discrete mechanics differs from continuous mechanics on the description of motion. In this respect,  a discrete motion is not a curve on the configuration manifold $Q$, it is rather a sequence of points on $Q$. 

We describe a variational discrete theory based on a discretized Hamilton's principle. From here we see that much of the theory evolves in parallel with the continuous Lagrangian theory. See \cite{marsden-west} for the main bibliographic account on the subject.

Let $L_d:Q\times Q\rightarrow\R$ be the \textit{discrete Lagrangian function}. Let us fix some  $N\in {\mathbb N}$ (number of steps) and a pair of points $q_0, q_N\in Q$

The \textit{discrete path space} is the space of sequences:  
$$C_d (q_0,q_N)= \{q_d\equiv \{q_k\}_{k=0}^{N} \ |\;  q_k \in Q \hbox{  and  } q_0, q_N \hbox{  fixed}  \}.$$

The \textit{discrete action map} is defined to be the map $S_d:C_d (q_0, q_N)\rightarrow \R$, 
\begin{equation}
S_d(q_d)=\sum_{k=0}^{N-1} L_d(q_k,q_{k+1}). \label{DA}
\end{equation}

Note that when one wishes to construct a numerical method using this approach, one usually regards the value of the discrete Lagrangian on a point $(q_0,q_1)$ as being an approximation of the (continuous) action, integrated over a solution connecting the two fixed points $q_0,q_1$ in a fixed time-step $h\in\R$, i.e.,
$$L_d (q_0,q_1)\approx \int_{0}^{h} L(q_{0,1}(t),\dot{q}_{0,1}(t)) \ dt,$$
where $L:TQ\rightarrow \R$ is a regular continuous-time Lagrangian function and $q_{0,1}(t)$ is the \textit{unique} solution of the Euler-Lagrange equations connecting $q_0$ and $q_1$ (as a consequence of Theorem \ref{Sodeexp}).

The \textit{discrete Hamilton's principle} states that a solution of the discrete Lagrangian system given by the discrete Lagrangian function $L_d$ is an extremum for the discrete action map \eqref{DA} among all sequences of points with fixed end-points. 
That is, $q_d\in C_{d}(q_0, q_N)$ is a solution if and only if $q_d$ is a critical point of the functional $S_d$, i.e. $$dS_d(q_d)(X_d) = 0,$$ for all $X_d \in T_{q_d} \mathcal{C}_d(q_0, q_N)$.

Analogously to the continuous-time case, we find out the \textit{discrete Euler-Lagrange equations} (DEL equations) as necessary and sufficient conditions to find extrema
\begin{equation}\label{DEL}
D_2 L_d(q_{k-1},q_{k})+D_1 L_d (q_{k},q_{k+1})=0, \ \ \text{for all} \  k=1,...,N-1.
\end{equation}
where $D_1L_d(q_{k-1},q_k)\in T^*_{q_{k-1}}Q$ and $D_2L_d(q_{k-1},q_k)\in T^*_{q_k}Q$ correspond to 
$dL_d(q_{k-1},q_k)$ under the identification $T^*_{(q_{k-1},q_k)}(Q\times Q)\cong T^*_{q_{k-1}}Q\times T^*_{q_k}Q$, that is,
\[
dL_d(q_{k-1},q_k)=D_1 L_d (q_{k-1},q_{k})+D_2 L_d (q_{k-1},q_{k})\; .
\]

Given a discrete Lagrangian $L_d: Q\times Q\rightarrow \R$ we can define two \emph{discrete Legendre transformations} $\F^{\pm} L_{d}:Q\times Q \rightarrow T^{*}Q$ given by
\begin{eqnarray*}
&& \mathbb{F}^{+}L_{d}(q_{k-1},q_{k})= (q_{k},D_2L_d(q_{k-1},q_{k})) \, ,\\
&& \mathbb{F}^{-}L_{d}(q_{k-1},q_{k}) = (q_{k-1},-D_1L_d(q_{k-1},q_{k})) \, .
\end{eqnarray*}
We say that $L_d$ if \textit{regular} if $\mathbb{F}^{+}L_{d}$ (or, equivalently, $\mathbb{F}^{-}L_{d}$ ) is a local diffeomorphism. This is equivalent to the regularity of the matrix $D_{12}L_d$.  

Under this regularity condition the 2- form on $Q\times Q$ defined by  $$(\mathbb{F}^{+}L_{d})^*\omega_Q=(\mathbb{F}^{-}L_{d})^*\omega_Q=:\Omega_{L_d}$$
is a symplectic form.

Moreover if  $L_d$ is regular then we can obtain a well defined \emph{discrete Lagrangian map}
\begin{equation*}
\begin{array}{cccc}
F_{L_d}: & Q\times Q & \longrightarrow & Q \times Q \\
& (q_{k-1},q_k) & \longmapsto & (q_k,q_{k+1}(q_{k-1},q_k)) \, ,
\end{array}
\end{equation*}
which is the discrete dynamical flow of our system. Here $q_{k+1}$ is the unique solution of the DEL equations \eqref{DEL} for the given pair $(q_{k-1},q_k)$. 
We can easily  check the symplecticity of the flow: 
\[
F_{L_d}^*\Omega_{L_d}=\Omega_{L_d}
\]

Alternatively, using the discrete Legendre transformations, we can also define the evolution of the discrete system on the cotangent bundle or \emph{Hamiltonian side}, $\widetilde{F}_{L_d}:T^*Q \longrightarrow T^*Q$, by any of the formulas
$$
\widetilde{F}_{L_d}=\mathbb{F}^{+}L_{d}\circ (\mathbb{F}^{-}L_{d})^{-1}=\mathbb{F}^{+}L_{d}\circ F_{L_d} \circ (\mathbb{F}^{+}L_{d})^{-1}=\mathbb{F}^{-}L_{d}\circ F_{L_d} \circ (\mathbb{F}^{-}L_{d})^{-1} \, ,
$$
because of the commutativity of the following diagram:
\begin{center}
	\begin{tikzcd}[column sep=tiny, row sep=huge]
		Q \times Q : & (q_{k-1},q_{k}) \arrow[dr, mapsto, "\mathbb{F}^{+}L_d"'] \arrow[rr, mapsto, "F_{L_d}"] & & (q_{k},q_{k+1}) \arrow[dl, mapsto, "\mathbb{F}^{-}L_d"'] \arrow[dr, mapsto, "\mathbb{F}^{+}L_d"] \arrow[rr, mapsto, "F_{L_d}"] & & (q_{k+1},q_{k+2}) \arrow[dl, mapsto, "\mathbb{F}^{-}L_d"]\\
		T^*Q :& & (q_k,p_k) \arrow[rr, mapsto, "\widetilde{F}_{L_d}"] & & (q_{k+1},p_{k+1}) & &
	\end{tikzcd}
\end{center}
The discrete Hamiltonian map $\widetilde{F}_{L_d}:(T^*Q,\omega_Q) \longrightarrow (T^*Q,\omega_Q)$ is symplectic  where $\omega_Q$ is the canonical symplectic 2-form on $T^*Q$.

If we start with a continuous Lagrangian and somehow derive an appropriate discrete Lagrangian, then the DEL equations become a geometric integrator for the continuous Euler-Lagrange system, known as a variational integrator.
This method to construct integrators for Lagrangian systems enjoys plenty of nice geometric features such as a symplectic discrete flow and discrete momentum conservation \cite{marsden-west}.

Hence, given a regular Lagrangian function $L: TQ \longrightarrow \mathbb{R}$, we define a discrete Lagrangian $L_d$ as an approximation of the action of the continuous Lagrangian. More precisely, for a regular Lagrangian $L$ and appropriate $h>0$, $q_0,q_1\in Q$, we can define the \textit{exact discrete Lagrangian} function $L_d^{e,h}:Q\times Q\rightarrow \R$ giving an exact correspondence between continuous and discrete motions as
\begin{equation}
L_d^{e,h}(q_0,q_1)=\int_{0}^{h} L(q_{0,1}(t),\dot{q}_{0,1}(t)) \ dt. \label{EDL}
\end{equation}
Again, $q_{0,1}(t)$ is the \textit{unique} solution of the Euler-Lagrange equations connecting $q_0$ and $q_1$ with $h$ small enough. Observe  that the solutions of Discrete Euler-Lagrange equations for $L$ exactly lie on the solutions of the Euler-Lagrange equations for $L_d^{e,h}$. In fact, in \cite{marsden-west}, the authors prove the following theorem which gives us the correspondence between discrete and continuous Lagrangian mechanics:
\begin{theorem}\label{exact:correspondence}
Take a series of times $\{t_k=kh, k=0,...,N\}$ for a sufficiently small time-step $h\in \R$, a regular Lagrangian $L$ and its corresponding discrete Lagrangian function $L_d^{e,h}$. Let $q(t)$ be a solution of Euler-Lagrange equations for $L$ satisfying the boundary conditions $q(0)=q_0$ and $q(t_N)=q_N$. Define a sequence $\{q_k\}_{k=0}^{N}$ in $Q$ by $$q_k=q(t_k), \ \ \text{for} \ k=0,...,N.$$ Then $\{q_k\}_{k=0}^{N}$ is a solution of the discrete Euler-Lagrange equations for $L_d^{e,h}\;$.

Conversely, if  we let $\{q_k\}_{k=0}^{N}$ be a solution of the discrete Euler-Lagrange equations for $L_d^{e,h}$, then the curve $q:[0,t_N]\rightarrow Q$ defined by $$q(t)=q_{k,k+1}(t), \ \ \text{for} \ t\in [t_k,t_{k+1}],$$ where $q_{k,k+1}(t)$ is the \textit{unique} solution of the Euler-Lagrange equations connecting $q_k$ and $q_{k+1}$, is a solution of Euler-Lagrange equations for $L$ on the whole interval $[0,t_N]$.
\end{theorem}

Following the Hamiltonian formalism, if we have a Hamiltonian problem defined by the Hamiltonian  $H = E_L \circ \left(\mathbb{F}L\right)^{-1}$,  then  the exact Hamiltonian map $\widetilde{F}_{L_d^{e,h}}$ coincides with the Hamiltonian flow $\phi_{h}^{X_{H}}$ of the continuous Hamiltonian system $H$ for a discrete amount of time $h$.
Now we recall the result of 
\cite{marsden-west} and \cite{PatrickCuell} for a discrete Lagrangian
$L_d\colon Q\times Q\rightarrow {\mathbb R}$.

\begin{definition}\label{order}
	Let $L_{d}\colon Q\times Q\rightarrow {\mathbb R}$ be a discrete Lagrangian. We say that
	$L_{d}$ is a discretization of order $r$ if there exist an
	open subset $U_{1}\subset TQ$ with compact closure and
	constants $C_1>0$, $h_1>0$ so that
	\begin{equation*}
	\lvert L_{d}(q(0),q(h))-L_{d}^{e,h}(q(0),q(h))\rvert\leq C_{1}h^{r+1}
	\end{equation*} for all solutions $q(t)$ of the second-order Euler--Lagrange equations with initial conditions $(q_0,\dot{q}_0)\in U_1$ and for all $h\leq h_1$.
\end{definition}
Following \cite{marsden-west,PatrickCuell}, we have the following important result about the order of a variational integrator.
\begin{theorem}\label{variational-error}
	If $\widetilde{F}_{L_d}$ is the evolution map of an order $r$
	discretization $L_d\colon Q\times Q\rightarrow {\mathbb R}$ of the exact discrete
	Lagrangian $L_d^{e,h}\colon Q\times Q\rightarrow {\mathbb R}$, then
	\[\widetilde{F}_{L_d}=\widetilde{F}_{L_{d}^{e,h}}+\mathcal{O}(h^{r+1}).\]
	In other words, $\widetilde{F}_{L_d}$ gives an integrator of order
	$r$ for $\widetilde{F}_{L_{d}^{e,h}}=\phi_{h}^{X_{H}}$.
\end{theorem}

This theorem gives us a method to find the order of a symplectic integrator for a mechanical system determined by a regular Lagrangian function $L: TQ\rightarrow \R$. We take a discrete Lagrangian $L_{d}\colon Q\times Q\to\R$ 
as an approximation of $L_d^{e,h}$ and the order can be calculated by expanding the expressions for
$L_d(q(0),q(h))$ in a Taylor series in $h$ and comparing this to the same expansions for the exact Lagrangian.
If the both series agree up to $r$ terms, then the discrete Lagrangian is of order $r$ (see \cite{marsden-west} and references therein).

\bigskip


\subsection{Forced discrete mechanics}\label{pl}

One of the most important properties of variational integrators is the possibility to adapt to more complex situations, for instance, systems involving forces or constraints (see \cite{marsden-west}). 

For the case of systems subjected to external forces, given a continuous force $F: TQ\rightarrow T^*Q$, we introduce  the discrete counterpart as two maps  $F_d^{+}: Q\times Q\longrightarrow T^*Q$ and $F_d^{-}: Q\times Q\longrightarrow T^*Q$ called the \textit{discrete force maps}. These discrete forces satisfy  $\pi_{Q}\circ F_{d}^{+}=\text{pr}_{2}$ and $\pi_{Q}\circ F_{d}^{-}=\text{pr}_{1}$, where $\pi_{Q}$ is the canonical projection of the cotangent bundle, and $\text{pr}_{1,2}:Q\times Q\longrightarrow Q$ are the canonical projections onto the first and second factors, respectively.

Now, the discrete equations of motion are derived from the \textit{discrete Lagrange-d'Alembert principle}:  
\begin{equation}
\delta S_{d}(q_{d})\cdot \delta q_{d}+\sum_{k=1}^{N-1} \left[F_d^{+}(q_{k-1},q_{k})+F_d^{-}(q_{k},q_{k+1}) \right]\cdot\delta q_{k}=0
\end{equation}
for all variations $\delta q_{k}$, with $\delta q_0=\delta q_N=0$. 

The \textit{forced Euler-Lagrange equations} are given by 
\begin{equation}\label{forced}
	D_2 L_d(q_{k-1},q_{k})+D_1 L_d (q_{k},q_{k+1})+F_d^{+}(q_{k-1},q_{k})+F_d^{-}(q_{k},q_{k+1})=0\; .
\end{equation}
which implicitly define a discrete forced Lagrangian map ${F}_{L^f_d}: Q\times Q\rightarrow Q\times Q$.

As in the unforced case, we can define the corresponding discrete Legendre transformations $\F^{f\pm}L_{d}:Q\times Q \rightarrow T^{*}Q$  given by
\begin{eqnarray*}
	&& \mathbb{F}^{f+}L_{d}(q_{k-1},q_{k})= (q_{k},D_2L_d(q_{k-1},q_{k})+F_d^+(q_{k-1}, q_k)) \, ,\\
	&& \mathbb{F}^{f-}L_{d}(q_{k-1},q_{k}) = (q_{k-1},-D_1L_d(q_{k-1},q_{k})-F_d^-(q_{k-1}, q_k)) \, .
\end{eqnarray*}
If the discrete forced system is regular, that is, the discrete Legendre transformations $\F^{f\pm}L_{d}$ are local diffeomorphisms then we have an explicit discrete forced Lagrangian map ${F}_{L^f_d}$ which is a local diffeomorphism. In addition, we may consider the discrete forced Hamiltonian map $\widetilde{F}_{L^f_d}: T^*Q\rightarrow T^*Q$ 
\[
\widetilde{F}_{L^f_d}=\mathbb{F}^{f\pm }L_{d}\circ {F}_{L^f_d}\circ   \left(\mathbb{F}^{f\pm }L_{d}\right)^{-1}\; .
\]

Now suppose that $(L,F)$ is a forced continuous Lagrangian system with regular Lagrangian function $L:TQ\rightarrow \R$ and an external force $F:TQ\rightarrow T^{*}Q$. Then, as we know (see Section \ref{forced:mechanics:section}), the dynamical vector field is a SODE $\Gamma_{(L,F)}$ on $TQ$ which is characterized by condition \eqref{fLvf}.

We will denote by $$\text{exp}_{h}^{\Gamma_{(L,F)}}:\mathcal{U}_{h}\subseteq TQ\rightarrow Q \times Q$$
the exponential map associated with $\Gamma_{(L,F)}$ for a sufficiently small positive number $h$. This map is a local diffeomorphism and so we may consider the exact retraction associated to it, which is its inverse map $R_{h,F}^{e-}$.

Using the flow $\phi_{h}^{\Gamma_{(L,F)}}$ of $\Gamma_{(L,F)}$ and the associated exact retraction we may introduce the forced exact discrete Lagrangian function $L_{d,F}^{e,h}:Q \times Q\rightarrow \R$ given by
\begin{equation*}
L_{d,F}^{e,h}(q_0, q_1)=\int_{0}^{h} \left( L\circ \phi_{t}^{\Gamma_{(L,F)}} \circ R_{h,F}^{e-}\right) (q_{0},q_{1}) \ dt,
\end{equation*}
and the double exact discrete force $F_{d}^{e,h}:Q \times Q\rightarrow T^{*}(Q\times Q)$ defined by
\begin{equation*}
	\langle F_d^{e,h}(q_0, q_1), (X_{q_0}, X_{q_1})\rangle
	=\int^h_0 \langle \left( F\circ \phi_{t}^{\Gamma_{(L,F)}} \circ R_{h,F}^{e-}\right) (q_{0},q_{1}), X_{0,1}(t)\rangle\; dt
\end{equation*}
where $X_{0,1}(t)=T_{(q_0, q_1)}(\tau_Q\circ \phi_t^{\Gamma_{(L,F)}}\circ R^{e-}_{h,F})(X_{q_0}, X_{q_1})$, for $(X_{q_0}, X_{q_1})\in T_{q_{0}}Q\times T_{q_{1}}Q$.

Then, the exact discrete force maps are just $F_d^{e,+}: Q\times Q\rightarrow T^*Q$ and $F_d^{e,-}: Q\times Q\rightarrow T^*Q$ given by
\begin{eqnarray*}
	\langle F_d^{e,+}(q_0, q_1), X_{q_{1}}\rangle&=& \langle F_d^{e,h}(q_0, q_1), (0_{q_0}, X_{q_{1}})\rangle\\
	\langle F_d^{e,-}(q_0, q_1), X_{q_{0}}\rangle&=& \langle F_d^{e,h}(q_0, q_1), (X_{q_{0}}, 0_{q_1}, \rangle.
\end{eqnarray*}
Note that if we denote by $q:Q\times Q\times [0,h]\rightarrow Q$ the function defined by
\begin{equation*}
	q(q_{0},q_{1},t)=q_{0,1}(t),
\end{equation*}
where $q_{0,1}:[0,h]\rightarrow Q$ is the solution of the forced Lagrangian system satisfying $q_{0,1}(0)=q_{0}$ and $q_{0,1}(h)=q_{1}$, then it is clear that
\begin{equation*}
q_{0,1}(t)=\left( \tau_{Q}\circ \phi_{t}^{\Gamma_{(L,F)}} \circ R_{h,F}^{e-}\right) (q_{0},q_{1}).
\end{equation*}
So, with this notation, the maps $L_{d,F}^{e,h}$, $F_d^{e,+}$ and $F_d^{e,-}$ may be written as follows
\begin{equation*}
L_{d,F}^{e,h}(q_0,q_1)=\int_{0}^{h} L(q_{0,1}(t),\dot{q}_{0,1}(t)) \ dt,
\end{equation*}
\begin{equation*}
F_d^{e,+}(q_0,q_1)=\int_{0}^{h} \langle F(q_{0,1}(t),\dot{q}_{0,1}(t)),\frac{\partial q_{0,1}}{\partial q_{1}} \rangle \ dt
\end{equation*}
and
\begin{equation*}
F_d^{e,-}(q_0,q_1)=\int_{0}^{h} \langle F(q_{0,1}(t),\dot{q}_{0,1}(t)),\frac{\partial q_{0,1}}{\partial q_{0}} \rangle \ dt,
\end{equation*}
where
\begin{equation*}
	\frac{\partial q_{0,1}}{\partial q_{1}}:T_{q_{1}}Q\rightarrow T_{q_{0,1}(t)}Q, \quad \text{and} \quad \frac{\partial q_{0,1}}{\partial q_{0}}:T_{q_{0}}Q\rightarrow T_{q_{0,1}(t)}Q
\end{equation*}
are given by
\begin{equation*}
	\langle \frac{\partial q_{0,1}}{\partial q_{1}},X_{q_{1}} \rangle=T_{(q_{0},q_{1},t)}q (0_{q_{0}},X_{q_{1}},0_{t}), \quad \langle \frac{\partial q_{0,1}}{\partial q_{0}},X_{q_{0}} \rangle=T_{(q_{0},q_{1},t)}q (X_{q_{0}},0_{q_{1}},0_{t}),
\end{equation*}
for $X_{q_{0}}\in T_{q_{0}}Q$ and $X_{q_{1}}\in T_{q_{1}}Q$.

Using the previous definitions, one may prove a forced version of Theorem \ref{exact:correspondence} (cf. \cite{marsden-west}). Moreover, in \cite{Sato}, the authors give a forced version of Theorem \ref{variational-error} using the variational order of the corresponding duplicated system.

In fact, we will need a useful Lemma from \cite{marsden-west} in Section \ref{discrete:nonholonomic}.

\begin{lemma}\label{exactlegendre}
	Let $(Q,L,F)$ be a forced Lagrangian problem with regular Lagrangian function $L$. The corresponding exact discrete Legendre transformations satisfy
	\begin{enumerate}
		\item $\F^{f+}L_{d,F}^{e,h}(q_{0},q_{1})=\F L (q_{0,1}(h),\dot{q}_{0,1}(h))$;
		\item $\F^{f-}L_{d,F}^{e,h}(q_{0},q_{1})=\F L (q_{0,1}(0),\dot{q}_{0,1}(0))$;
	\end{enumerate}
	where $q_{0,1}(t)$ is the solution of the forced Euler-Lagrange equations verifying $q_{0,1}(0)=q_{0}$ and $q_{0,1}(h)=q_{1}$.
\end{lemma}


\section{Discrete nonholonomic mechanics}\label{discrete:nonholonomic}

In this section, we introduce a modification of the Lagrange-d'Alembert principle. Later, using the construction of the nonholonomic exponential map in Section \ref{sec:exponential-map}, we will define the exact discrete version of nonholonomic mechanics and show that it satisfies the modified Lagrange-d'Alembert principle.  

Let $\D$ be a distribution on the manifold $Q$. Let $L_{d}:Q\times Q\longrightarrow\R$ be a discrete Lagrangian function, $F_d^{\pm}:Q\times Q\longrightarrow T^*Q$ discrete forces and $\M_{d}\subseteq Q\times Q$ a discrete constraint space. We remark that $\pi_{Q}\circ F_{d}^{+}=\text{pr}_{2}$ and $\pi_{Q}\circ F_{d}^{-}=\text{pr}_{1}$, where $\pi_{Q}:T^{*}Q\rightarrow Q$ and $\text{pr}_{1,2}:Q\times Q\rightarrow Q$ are the canonical projections.

\begin{definition}\label{mLdA}
	A sequence $(q_{0}, ..., q_{N})$ in $Q$ satisfies the \textit{modified Lagrange-d'Alembert principle} if it extremizes
	\begin{equation}
	\begin{split}
	& dS_{d}(q_{d})\cdot \delta q_{d}+\sum_{k=1}^{N-1} \left[F^{+}(q_{k-1},q_{k})+F^{-}(q_{k},q_{k+1}) \right]\cdot\delta q_{k}=0 \\
	& (q_{k},q_{k+1})\in \M_{d}, \ 0\leqslant k \leqslant N-1
	\end{split}
	\end{equation}
	for all variations lying in the distribution $\delta q_{k}\in \D_{q_k}$, $\delta q_{d}=(\delta q_{0},...,\delta q_{N})\in T_{q_{d}}\C_{d}(q_{0},q_{N})$ and $\delta q_0=\delta q_N=0$.
\end{definition}

\begin{remark}{\rm 
		Observe that this principle is exactly the same that discrete Lagrange-d'Alembert principle for forced systems when $\D=TQ$ and $\M_{d}=Q\times Q$. It is also the  discrete Lagrange-d'Alembert principle for nonholonomic systems introduced by \cite{CM2001} when $F^+=F^-=0$.
		Also,  in this context we find the methods proposed by \cite{MR2134330}, using a discretization of the forces for a nonholonomic system and a discrete submanifold derived from the continuous constraints and the forced discrete Legendre transformations  
		Recently, a similar principle was introduced in \cite{parks} to study discretizations of Dirac mechanics. 
	}
\end{remark}

Now, as in the case of forced systems, we have that
\begin{proposition}
	A sequence $(q_{0}, ..., q_{N})$ in $Q$ satisfies the modified Lagrange-d'Alembert principle if and only if it satisfies \textit{modified Lagrange-d'Alembert equations}
	\begin{eqnarray}\label{MLA}
	&& D_2 L_d(q_{k-1},q_{k})+D_1 L_d (q_{k},q_{k+1})+F^{+}(q_{k-1},q_{k})+F^{-}(q_{k},q_{k+1}) \in \D^{o}_{q_k} \nonumber \\
	&& \omega^{a}(q_{k},q_{k+1})=0, \ 0\leqslant k \leqslant N-1,
	\end{eqnarray}
	where $\M_d$ is determined by the zeros of a set of  constraint functions $\omega^{a}:Q\times Q\longrightarrow \R$.
\end{proposition}

\subsection{Exact discrete nonholonomic equations}\label{discrete:nonholonomic-1}

If we denote the inclusion of $\D$ in $TQ$ by $i_{\D}: \D\hookrightarrow TQ$, we induce   the dual projection 
$i_\D^*:  T^*Q\rightarrow \D^*$  defined by
\begin{equation*}
	\langle i_D^*(\mu_q), v_q\rangle =\langle \mu_q, i_{\D}(v_q)\rangle, \quad \mu_q\in T^*_qQ, \ v_q\in \D_q.
\end{equation*} 

The Legendre transformations of the Lagrangian functions $L:TQ\rightarrow \R$ and $l=L|_{\D}: \D\rightarrow \R$ satisfy the following relation
\begin{equation}\label{legendre}
	i_{\D}^*\circ \F L\circ i_{\D}=\F l,
\end{equation}
where $\F l: \D\rightarrow \D^*$ is the restricted Legendre transformation defined from $l$ (see Subsection \ref{nho}).

Now consider the {\sl exact discrete nonholonomic Legendre transformations}
$\F^{\pm}_{h,nh} l: \M_{h}^{e,nh}\rightarrow \D^*$ defined by 
\begin{eqnarray*}
\F^{-}_{h,nh} l (q_0, q_1)&=&\F l\circ R_{h,nh}^{e-}(q_0,q_1)\in \D^*_{q_0}\\
\F^{+}_{h,nh} l (q_0, q_1)&=&\F l \circ R_{h,nh}^{e+}(q_0,q_1)\in \D^*_{q_1}.
\end{eqnarray*}
Note that $\F^{\pm}_{h,nh} l$ are (local) diffeomorphisms.

As we will see below, the condition of momentum matching gives the {\sl exact discrete  nonholonomic equations}:   
\begin{equation}\label{momemtummatch}
	\begin{split}
	\F^{+}_{h, nh} l (q_0, q_1)-\F^{-}_{h,nh} l (q_1, q_2) & = 0 \\
	(q_0, q_1), (q_{1},q_{2}) & \in \M_{h}^{e,nh}.
	\end{split}
\end{equation}
We sill see in a theorem below why they are called "exact".

\begin{remark}
	{\rm 
Alternatively we can define the subset 
\begin{equation*}
S_{nh}^e = \{ (\F l\circ R_{h, nh}^{e-}(q_0,q_1), \F l\circ R_{h, nh}^{e+}(q_0,q_1)) \ | \ (q_0, q_1)\in M_{h}^{e,nh}\}
\end{equation*}
and we can think $S_{nh}^e\subset \D^*\times \D^*$ as an implicit difference equation \cite{IMMP} producing the exact discrete nonholonomic dynamics.}
\end{remark}

Observe that, since both $R_{h,nh}^{e-}$ and $\F l$ are local diffeomorphisms, then equations \eqref{momemtummatch} implicitly define an \textit{exact discrete flow}: $\Phi_{h,nh}^e:\M_{h}^{e,nh}\rightarrow \M_{h}^{e,nh}$ by
\begin{equation}\label{exactflow}
	\Phi_{h,nh}^e (q_0, q_1)= \text{exp}_{h}^{\Gamma_{nh}}\circ R_{h,nh}^{e+}(q_0,q_1).
\end{equation}

Moreover,  it produces a well-defined flow on $\D^*$, denoted by $\varphi^e_{h,nh}: \D^*\rightarrow \D^*$, which is defined  by
\begin{equation*}
	\varphi_{h,nh}^e(\mu_{q_0})= \F^{+}_{h,nh} l \circ ( \F^{-}_{h,nh} l)^{-1}(\mu_{q_0}), \quad \mu_{q_0}\in D^*_{q_0}.
\end{equation*}
The interplay between both discrete flows and the nonholonomic Legendre transformations may be summarized in the following commutative diagram (see Figure \ref{dis-flow-non-Le-trans}).
\begin{figure}[htb!]
	\centering
	\caption{Commutative diagram. Exact discrete and continuous noholonomic flows}
	\label{dis-flow-non-Le-trans}
	\begin{tikzcd}
	& \mathcal{M}_{h}^{e,nh} \arrow[rr, "\Phi_{h,nh}^{e}"] \arrow[ld, "\mathbb{F}_{h,nh}^{-}l"'] \arrow[rd, "\mathbb{F}_{h,nh}^{+}l"] &                 & \mathcal{M}_{h}^{e,nh} \arrow[ld, "\mathbb{F}_{h,nh}^{-}l"] \\
	\mathcal{D}^{*} \arrow[rr, "\varphi_{h,nh}^{e}"'] &                                                                                                                             & \mathcal{D}^{*} &                                                          
	\end{tikzcd}
\end{figure}

Having the construction of nonholonomic integrators in mind, it is interesting to observe that the exact discrete nonholonomic dynamics  exactly reproduces the continuous flow of the nonholonomic system at any step $h$.

\begin{theorem}
	Given $(q_{0},q_{1})\in \M_{h}^{e,nh}$ and $h>0$, consider the sequence $(q_{0},q_{1},...,q_{N})$ obtained by multiple iterations of the exact discrete flow $\Phi^e_{h,nh}$ and thus, by definition, satisfying the exact discrete  nonholonomic equations
	\begin{equation}\label{ddd}
		\F^{+}_{h,nh} l (q_{k-1}, q_{k})-\F^{-}_{h,nh} l (q_{k}, q_{k+1})=0, \quad (q_{k},q_{k+1})\in \M_{h}^{e,nh}, 
	\end{equation}
for $0\leqslant k \leqslant N-1$.

	Then, we have that:
	\begin{enumerate}
		\item The sequence $(q_{0},q_{1},...,q_{N})$ exactly matches the trajectories of $\Gamma_{nh}$ in the sense that
			\begin{equation}
				q_{k}=q_{0,1}(kh),
			\end{equation}
			where $q_{0,1}$ is the unique trajectory of $\Gamma_{nh}$ satisfying $q_{0,1}(0)=q_{0}$ and $q_{0,1}(h)=q_{1}$.
		\item The Legendre transforms satisfy the equation
		\begin{equation}
		\F^{+}_{h, nh}l(q_0, q_1) =\phi^{\bar{\Gamma}_{nh}}_{h}( \F^{-}_{h, nh}l(q_0, q_1)) \; ,
		\end{equation}
		where $\{\phi^{\bar{\Gamma}_{nh}}_{h}\}$ is the flow of the vector field $\bar{\Gamma}_{nh}=(\F l)_* \Gamma_{nh}\in {\mathfrak X}(\D^*)$, that is  $\phi^{\bar{\Gamma}_{nh}}_{h}=\varphi_{h,nh}^{e}$.
	\end{enumerate}
\end{theorem}

\begin{proof}
The first item is a direct consequence of the definition of the exact discrete flow in \eqref{exactflow}. The second item is just a consequence of the definition of the Hamiltonian vector field as being $\F l$-related to the non-holonomic vector field $\Gamma_{nh}$. Indeed
\begin{equation*}
	\F^{+}_{nh}l=\F l \circ R_{h,nh}^{e+}=\F l \circ \phi^{\Gamma_{nh}}_{h} \circ R_{h,nh}^{e-}=\phi^{\bar{\Gamma}_{nh}}_{h} \circ \F l \circ R_{h,nh}^{e-}=\phi^{\bar{\Gamma}_{nh}}_{h}\circ \F^{-}_{nh}l.
\end{equation*}
\end{proof}

For the construction  of geometric integrators we will need another alternative expression of Equations (\ref{ddd}). 
In particular, using \eqref{legendre} we can rewrite these equations in a way that are very similar to the modified Lagrange-d'Alembert equations defined in Equation (\ref{MLA}) as
\begin{eqnarray*}
	i^*_{D}\left((\F L \circ i_{\D} \circ R_{h,nh}^{e+}) (q_0,q_1)-(\F L \circ i_{\D} \circ R_{h,nh}^{e-})(q_1,q_2)\right)&=&0\\
	(q_{0},q_{1}),(q_1, q_2)&\in& \M_{h}^{e,nh}
\end{eqnarray*}
or in other words
\begin{equation}\label{nhintegrator}
	\begin{split}
		(\F L\circ R_{h,nh}^{e+}(q_{0},q_1)-\F L\circ R_{h,nh}^{e-})(q_1,q_{2}) & \in\D^{o}_{q_1}\\
		(q_{0},q_{1}),(q_1, q_{2}) & \in \M_{h}^{e,nh},
	\end{split}
\end{equation}
where we omit $i_{\D}$ since $R_{h,nh}^{e+}(q_{0},q_1)$ and $R_{h,nh}^{e-}(q_1,q_{2})$ are vectors in the distribution $\D$ and may be identified with its inclusion.

\subsection{Exact discrete nonholonomic equations from a modified Lagrange-d'Alembert principle}

Given a regular nonholonomic system determined by the triple $(Q, L, D)$, we have seen how to derive the nonholonomic force $F_{nh}: \D\rightarrow T^*Q$ by modifying the free dynamics to satisfy the nonholonomic constraints.
 
Consider now an arbitrary extension $\widetilde{F_{nh}}: T Q\rightarrow T^*Q$ of $F_{nh}$. It is clear that the solutions of the forced system determined by $(L, \widetilde{F_{nh}})$ with initial conditions in $\D$, remain in $\D$ and match the trajectories of the nonholonomic system. In fact, if $\Gamma_{nh}$ is the nonholonomic dynamics and $\Gamma_{(L,\widetilde{F_{nh}})}$ is the forced dynamics, then it is clear that $\Gamma_{nh}=\Gamma_{(L,\widetilde{F_{nh}})}|_{\D}$.

If $R_{h,\widetilde{F_{nh}}}^{e-}$ is the exact retraction associated with the forced SODE $\Gamma_{(L,\widetilde{F_{nh}})}$ then, as in Section \ref{pl}, we may define the exact discrete versions
\begin{equation*}
L_{d,\widetilde{F_{nh}}}^{e,h}(q_0, q_1)=\int_{0}^{h} \left( L\circ \phi_{t}^{\Gamma_{(L,\widetilde{F_{nh}})}}\circ R_{h,\widetilde{F_{nh}}}^{e-}\right) (q_{0},q_{1}) \ dt,
\end{equation*}
and
\begin{equation*}
	\begin{split}
		\langle (\widetilde{F_{nh}})_d^{e,+}(q_0, q_1), X_{q_{1}}\rangle & = \langle F_d^{e,h}(q_0, q_1), (0_{q_0}, X_{q_{1}})\rangle \\
		\langle (\widetilde{F_{nh}})_d^{e,-}(q_0, q_1), X_{q_{0}}\rangle & = \langle F_d^{e,h}(q_0, q_1), (X_{q_{0}}, 0_{q_1}, \rangle,
	\end{split}
\end{equation*}
where $F_d^{e,h}:Q\times Q\rightarrow T^{*}(Q\times Q)$ is the double exact discrete force given by
\begin{equation*}
\langle F_d^{e,h}(q_0, q_1), (X_{q_0}, X_{q_1})\rangle
=\int^h_0 \langle \left( \widetilde{F_{nh}}\circ \phi_{t}^{\Gamma_{(L,\widetilde{F_{nh}}})} \circ R_{h,\widetilde{F_{nh}}}^{e-}\right) (q_{0},q_{1}), X_{0,1}(t)\rangle\; dt
\end{equation*}
where $X_{0,1}(t)=T_{(q_0, q_1)}(\tau_Q\circ \phi_t^{\Gamma_{(L,\widetilde{F_{nh}})}}\circ R^{e-}_{h,\widetilde{F_{nh}}})(X_{q_0}, X_{q_1})$, for $(X_{q_0}, X_{q_1})\in T_{q_{0}}Q\times T_{q_{1}}Q$.

Following the notation in \cite{marsden-west}, we may rewrite these maps as
\begin{eqnarray*}
		L_{d,\widetilde{F_{nh}}}^{e,h}(q_0, q_1)&=&\int_{0}^{h} L(q_{0,1}(t),\dot{q}_{0,1}(t)) \ dt\; ,  \\
		(\widetilde{F_{nh}})^{e,+}_d(q_0, q_1)&=& \int_{0}^{h} \langle (\widetilde{F_{nh}})(q_{0,1}(t),\dot{q}_{0,1}(t)), \frac{\partial q_{0,1}(t)}{\partial q_1}\rangle\ dt\; ,\\
		(\widetilde{F_{nh}})^{e,-}_d(q_0, q_1)&=& \int_{0}^{h} \langle (\widetilde{F_{nh}})(q_{0,1}(t),\dot{q}_{0,1}(t)), \frac{\partial q_{0,1}(t)}{\partial q_0}\rangle\ dt\; .
	\end{eqnarray*}
where now $q_{0,1}: [0, h]\rightarrow Q$ is the solution of the forced Euler-Lagrange equations for $(L, \widetilde{F_{nh}})$  verifying $q_{0,1}(0)=q_{0}$ and $q_{0,1}(h)=q_{1}$. 

We now prove that when we apply the modified Lagrange-d'Alembert principle to the exact discrete objects defined above, we obtain the exact discrete  nonholonomic equations.

\begin{theorem}
	Let $(Q,L,\D)$ be a regular continuous-time nonholonomic problem with regular Lagrangian $L$. Consider the exact discrete Lagrangian function $L_{d,\widetilde{F_{nh}}}^{e,h}$ defined above, as well as the exact discrete forces $(\widetilde{F_{nh}})^{e,-}_d$ and $(\widetilde{F_{nh}})^{e,+}_d$. Also let $\M_{h}^{e,nh}$ be the exact discrete space associated to $(Q,L,\D)$. Then the modified Lagrange-d'Alembert principle induces modified Lagrange-d'Alembert equations
	\begin{equation}\label{MLAnh}
	\begin{split}
	& D_2 L_{d,\widetilde{F_{nh}}}^{e,h}(q_{k-1},q_{k})+D_1 L_{d,\widetilde{F_{nh}}}^{e,h} (q_{k},q_{k+1})\\
	&+(\widetilde{F_{nh}})^{e,+}_d(q_{k-1},q_{k})+(\widetilde{F_{nh}})^{e,-}_d(q_{k},q_{k+1}) \in \D^{o}_{q_k} \\
	& (q_{k},q_{k+1})\in \M_{h}^{e,nh}, \ 0\leqslant k \leqslant N-1,
	\end{split}
	\end{equation}
	which are equivalent to the exact discrete  nonholonomic equations \eqref{momemtummatch}.
\end{theorem}

\begin{proof}
	The terms appearing in equations \eqref{MLAnh} are the restriction to $\M_{h}^{e,nh}$ of the exact discrete Legendre transformations for the forced system $(Q,L,\widetilde{F_{nh}})$:
	\begin{equation*}
	 \begin{split}
	 & \F^{f+}L_{d,\widetilde{F_{nh}}}^{e,h}(q_{k-1},q_{k})-\F^{f-}L_{d,\widetilde{F_{nh}}}^{e,h}(q_{k},q_{k+1}) \in \D^{o}_{q_k} \\
	 & (q_{k},q_{k+1})\in \M_{h}^{e,nh}, \ 0\leqslant k \leqslant N-1.
	 \end{split}
	 \end{equation*}
	 Thus, using Lemma \ref{exactlegendre}, the equations above are equivalent to
	 \begin{equation}\label{auxMLA}
	 \begin{split}
	 & \F L \circ R_{h, \widetilde{F_{nh}}}^{e,+}(q_{k-1},q_{k})-\F L \circ R_{h,\widetilde{F_{nh}}}^{e,-} (q_{k},q_{k+1}) \in \D^{o}_{q_k} \\
	 & (q_{k},q_{k+1})\in \M_{h}^{e,nh}, \ 0\leqslant k \leqslant N-1.
	 \end{split}
	 \end{equation}
	 Observe that, since the restriction of the forced dynamics to $\D$ matches the nonholonomic dynamics, then also the restriction of the forced retractions maps to $\M_{h}^{e,nh}$ matches the nonholonomic retraction maps $R_{h,nh}^{e,\pm}$.
	 
	 Now, if the sequence $(q_{0},...,q_{N})$ satisfies equations \eqref{momemtummatch}, then, since $\F L$ is a diffeomorphism one has that
	 \begin{equation*}
	 \begin{split}
	 & R_{h, nh}^{e,+}(q_{k-1},q_{k})=R_{h, nh}^{e,-} (q_{k},q_{k+1}) \\
	 & (q_{k},q_{k+1})\in \M_{h}^{e,nh}, \ 0\leqslant k \leqslant N-1,
	 \end{split}
	 \end{equation*}
	 and therefore equations \eqref{auxMLA} is trivially satisfied.
	 
	 Conversely, if the sequence $(q_{0},...,q_{N})$ satisfies equations \eqref{auxMLA}, then projecting by $i_{\D}^{*}$ we obtain \eqref{momemtummatch}.
\end{proof}

Observe that, we are restricting to pairs of points in $\M_d^{e,nh}$ and applying the modified Lagrange-d'Alembert principle 
\begin{eqnarray*}
	&&dS_{d} (q_{d})\cdot \delta q_{d} +\sum_{k=1}^{N-1}  \left[(\widetilde{F_{nh}})^{e,+}_d (q_{k-1}, q_k)+
	(\widetilde{F_{nh}})^{e,-}_d (q_{k}, q_{k+1})\right]\delta q_k=0\\
	&&(q_{k}, q_{k+1})\in \M_h^{e,nh},
\end{eqnarray*}
with $\delta q_{d}=(\delta q_{0},...,\delta q_{N})$ for all variations $\delta q_k\in \D_{q_k}$ verifying $\delta q_0=\delta q_N=0$ and
\begin{equation*}
	S_{d} (q_{d})=\sum_{k=0}^{N-1} L_{d,\widetilde{F_{nh}}}^{e,h}(q_{k},q_{k+1}).
\end{equation*}

\subsection{Construction of non-holonomic integrators}

To construct variational integrators we consider  discretizations $(L_d, F_d^-, F_d^+)$ of 
$(L_{d,\widetilde{F_{nh}}}^{e,h}, (\widetilde{F_{nh}})^{e,-}_d, (\widetilde{F_{nh}})^{e,+}_d)$ as a typical forced integrator and then we consider a discretization  $\M_{h}^{d}$ of $\M^{e,nh}_{h}$ to derive the {\it modified discrete Lagrange-d'Alembert equations}:  
\begin{equation}\label{MLAnh-d}
\begin{split}
& D_2 L_d(q_{k-1},q_{k})+D_1 L_d (q_{k},q_{k+1})+F^+_d(q_{k-1},q_{k})+F^-_d(q_{k},q_{k+1}) \in \D^{o}_{q_k} \\
& (q_{k},q_{k+1})\in \M_h^d, \ 0\leqslant k \leqslant N-1,
\end{split}
\end{equation}

We remark that \eqref{MLAnh-d} is equivalent to the projection onto $\D^{*}$, i.e.,
\begin{equation}\label{MLAnh-d2}
\begin{split}
& i_{\D}^{*}\left( D_2 L_d(q_{k-1},q_{k})+D_1 L_d (q_{k},q_{k+1})+F^+_d(q_{k-1},q_{k})+F^-_d(q_{k},q_{k+1}) \right)=0 \\
& (q_{k},q_{k+1})\in \M_h^d, \ 0\leqslant k \leqslant N-1,
\end{split}
\end{equation}
This projection motivates the definition of the Legendre transforms $\F^{\pm}l_{d}:\M^{d}_h\rightarrow \D^{*}$ given by
\begin{equation*}
\begin{split}
& \F^{+}l_{d}=i_{\D}^{*}\circ \F^{f+} L_{d}|_{\M^{d}_h} \\
& \F^{-}l_{d}=i_{\D}^{*}\circ \F^{f-} L_{d}|_{\M^{d}_h}.
\end{split}
\end{equation*}

\begin{example}
	Consider once more the nonholonomic particle. We introduce a discretization of the discrete space $\M^{e,nh}_{h}$
	\begin{equation}\label{discrete:space}
	\M^{d}_h=\{ z_{1}=z_{0}+\left( \frac{y_1+y_0}{2} \right)(x_1-x_0) \},
	\end{equation}
	and a discrete Lagrangian
	\begin{equation*}
	L_{d}(x_0,y_0,z_0,x_1,y_1,z_1)=\frac{1}{2h}\left[ \left(x_1-x_0 \right)^2+\left( y_1-y_0 \right)^2 +\left( z_1-z_0 \right)^2 \right].
	\end{equation*}
	Moreover we need two discrete forces
	\begin{equation*}
		F^+_d(q_{0},q_{1})=\frac{2}{h}\frac{(x_1-x_0)(y_1-y_0)}{4+\left( y_1+y_0 \right)^{2}}\left( -\frac{y_1+y_0}{2}d x_1+d z_1 \right)
	\end{equation*}
	and
	\begin{equation*}
		F^{-}_d(q_{0},q_{1})=\frac{2}{h}\frac{(x_1-x_0)(y_1-y_0)}{4+\left( y_1+y_0 \right)^{2}}\left( -\frac{y_1+y_0}{2}d x_{0}+d z_{0} \right).
	\end{equation*}
	The forced discrete Legendre transforms which appear also in the modified Lagrange-d'Alembert equations are
	\begin{equation*}
	\begin{split}
	\F^{f-} L_{d}(q_{0},q_{1})=\left( \frac{x_1-x_0}{h}+\frac{1}{h}\frac{(x_1-x_0)(y_1-y_0)(y_1+y_0)}{4+\left( y_1+y_0 \right)^{2}} \right)dx_{0} \\
	+\frac{y_1-y_0}{h}dy_{0} + \left( \frac{z_1-z_0}{h}-\frac{2}{h}\frac{(x_1-x_0)(y_1-y_0)}{4+\left( y_1+y_0 \right)^{2}} \right)dz_{0}
	\end{split}
	\end{equation*}
	and
	\begin{equation*}
	\begin{split}
	\F^{f+} L_{d}(q_{0},q_{1})=\left( \frac{x_1-x_0}{h}-\frac{1}{h}\frac{(x_1-x_0)(y_1-y_0)(y_1+y_0)}{4+\left( y_1+y_0 \right)^{2}} \right)dx_{1} \\
	+\frac{y_1-y_0}{h}dy_{1} + \left( \frac{z_1-z_0}{h}+\frac{2}{h}\frac{(x_1-x_0)(y_1-y_0)}{4+\left( y_1+y_0 \right)^{2}} \right)dz_{1}.
	\end{split}
	\end{equation*}
	Now projecting the forced Legendre transforms onto $\D^{*}$ by means of $i_{\D}^{*}$ and restricting to $\M^{d}_h$ we get
	\begin{equation*}
	\F^{-}l_{d}(q_{0}^{i},q_{1}^{a})=\frac{x_{1}-x_{0}}{h}\left( 1+\frac{1}{2}y_{0}(y_{1}+y_{0})+\frac{(y_{1}-y_{0})^{2}}{4+( y_1+y_0 )^{2}} \right)e^{1}+\left( \frac{y_1-y_0}{h} \right)e^{2}
	\end{equation*}
	and
	\begin{equation*}
	\F^{+}l_{d}(q_{0}^{i},q_{1}^{a})=\frac{x_{1}-x_{0}}{h}\left( 1+\frac{1}{2}y_{1}(y_{1}+y_{0})+\frac{(y_{1}-y_{0})^{2}}{4+( y_1+y_0 )^{2}} \right)e^{1}+\left( \frac{y_1-y_0}{h} \right)e^{2},
	\end{equation*}
	where the local frame $\{e^{a}\}\subseteq \D^{*}$ is dual to the local frame $\{e_{a}\}$ spanning $\D$, where $e_{1}=\frac{\partial}{\partial x}+y\frac{\partial}{\partial z}$ and $e_{2}=\frac{\partial}{\partial y}$.
	
	Now solving equations \eqref{MLAnh-d2} for this example we get
	\begin{equation*}
	\begin{split}
	& x_{2}=x_{1}+(x_{1}-x_{0})\frac{1+\frac{1}{2}y_{1}(y_{1}+y_{0})+\frac{(y_{1}-y_{0})^{2}}{4+( y_1+y_0 )^{2}}}{1+\frac{1}{2}y_{1}(3y_{1}-y_{0})+\frac{(y_{1}-y_{0})^{2}}{4+( 3y_1-y_0 )^{2}}} \\
	& y_{2}=2y_{1}-y_{0}.
	\end{split}
	\end{equation*}
	
	We can see in Figures \ref{fig:testa} and \ref{fig:testb} a comparison between the proposed integrator (MLA) and the more standard Discrete Lagrange-d'Alembert (DLA) integrator. We compare the error in both integrators as well as the energy behaviour of both. We observe the proposed integrator as good behaviour in both aspects and it even behaves slightly better than DLA. Notice that the Hamiltonian function $H|_{\D^{*}}$ given by
	\begin{equation*}
	H|_{\D^{*}}(x,y,z,p_{1},p_{2})=\frac{1}{2}\left( \frac{p_{1}^{2}}{1+y^{2}}+p_{2}^{2} \right)
	\end{equation*}
	becomes constant along the discrete flow, after the first steps. To run the simulation we set the initial position at the origin $q_{0}=0$ and $q_{1}=(0.4,0.4,z_{1})$, with $z_{1}$ being determined by \eqref{discrete:space}. The step is $h=0.5$ and the total number of steps is $N=1200$.
	
	\begin{figure}[htb!]
	\hspace{-1.5cm}	\includegraphics[width=1.2\linewidth]{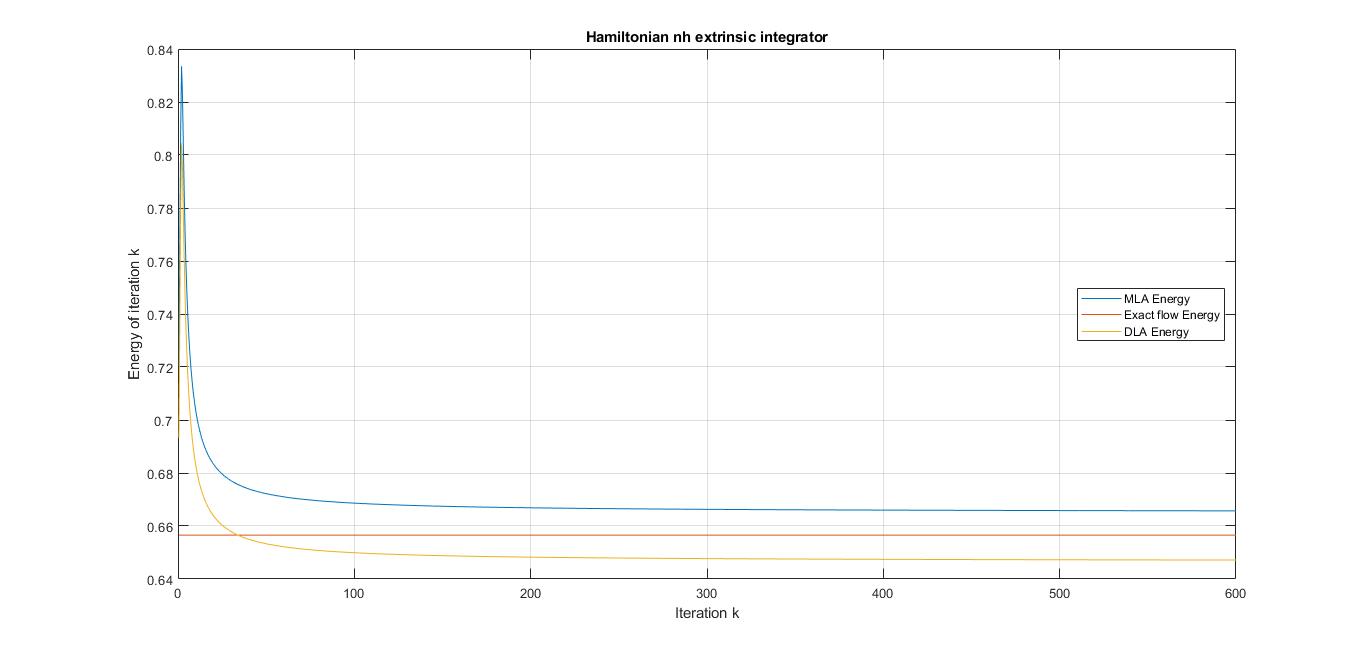}
		\caption{Comparison of the value of the Hamiltonian function between DLA and MLA integrators.}
		\label{fig:testa}
	\end{figure}
	\begin{figure}[htb!]
	\hspace{-1.5cm}	\includegraphics[width=1.2\linewidth]{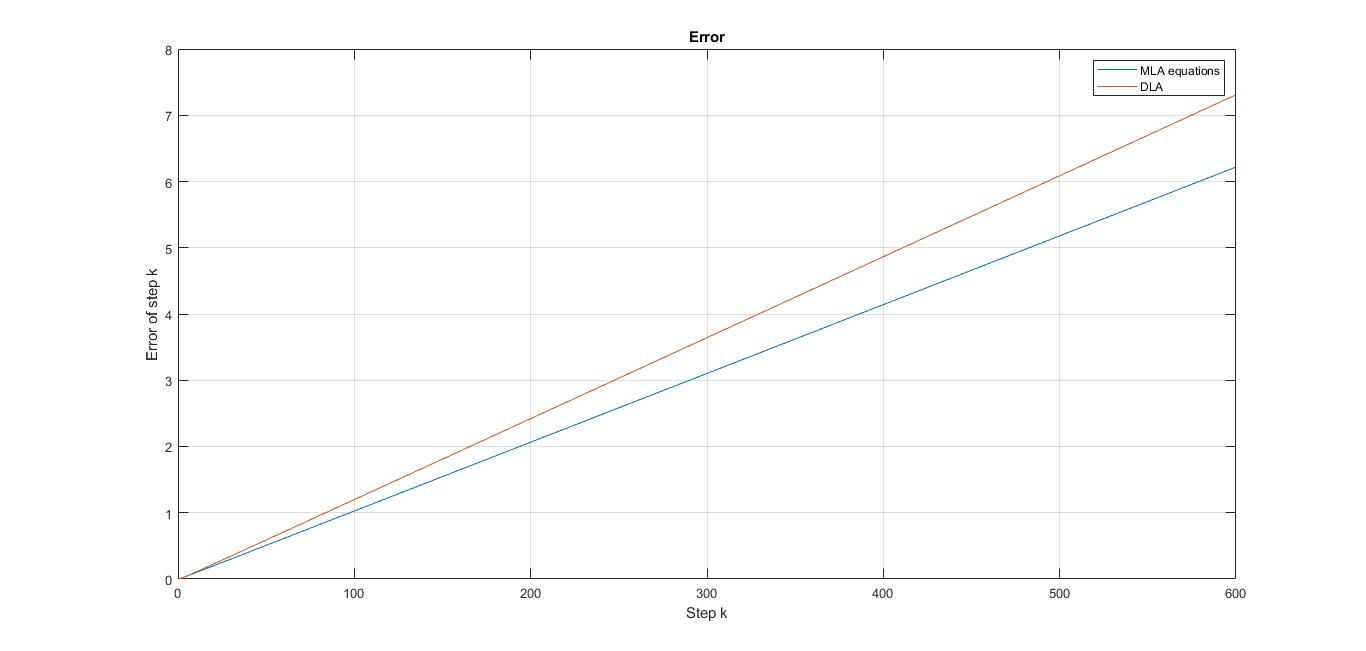}
		\caption{Evolution of the error in DLA and MLA integrators.}
		\label{fig:testb}
	\end{figure}
	
	We also draw in Figure \ref{fig:test2} the discrete constraint space $\M^{d}_h$ and compare it with its exact version $\M^{e,nh}_{h}$.
	\begin{figure}[htb!]
		\centering
		\begin{subfigure}{.5\textwidth}
			\begin{flushleft}
				\includegraphics[width=1.1\linewidth]{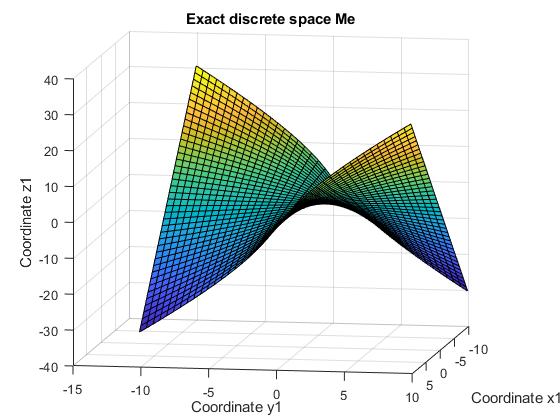}
				\caption{Exact discrete space $\M^{e,nh}_{h}$ given by \eqref{defdiscrete}.}
				\label{fig:sub1}
			\end{flushleft}
		\end{subfigure}%
		\begin{subfigure}{.5\textwidth}
			\centering
			\includegraphics[width=1.1\linewidth]{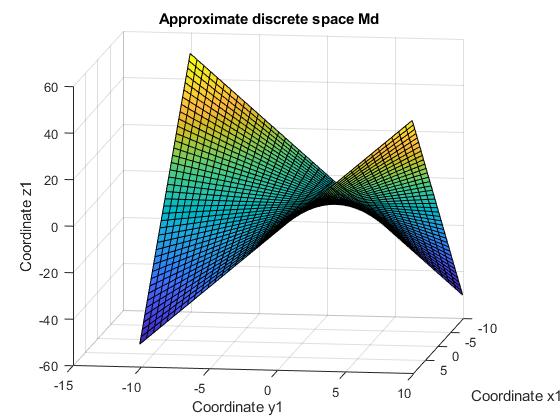}
			\caption{Discrete space $\M_{h}^{d}$ given by \eqref{discrete:space}.}
			\label{fig:sub2}
		\end{subfigure}
		\caption{Graph of the defining function for the respective spaces. We have fixed the origin as the initial point $q_{0}=0$ and plotted the coordinate $z_{1}$ as a function of $x_{1}$ and $y_{1}$.}
		\label{fig:test2}
	\end{figure}
\end{example}

\begin{example}
	Let us introduce another typical example of nonholonomic system (see \cite{Bloch}): the knife edge. Choosing appropriate constants, its Lagrangian function is described by the function $L:T(Q\times \Es^{1})\rightarrow \R$
	\begin{equation*}
	L(x,y,\varphi,\dot{x},\dot{y},\dot{\varphi})=\frac{1}{2}(\dot{x}^{2}+\dot{y}^{2}+\dot{\varphi}^{2})+\frac{x}{2},
	\end{equation*}
	and it is subjected to the nonholonomic constraint
	\begin{equation*}
		\sin(\varphi)\dot{x}-\cos(\varphi)\dot{y}=0.
	\end{equation*}
	We introduce the following discretization of the constraint space
	\begin{equation*}
		\M^{d}_h=\left\{\sin\left(\frac{\varphi_{1}+\varphi_{0}}{2}\right)\frac{x_{1}-x_{0}}{h}-\cos\left(\frac{\varphi_{1}+\varphi_{0}}{2}\right)\frac{y_{1}-y_{0}}{h}=0\right\}.
	\end{equation*}
	The natural discretization of the Lagrangian compatible with the above discrete constraint space is then
	\begin{equation*}
		L_{d}(x_{0},y_{0},\varphi_{0},x_{1},y_{1},\varphi_{1})=\frac{1}{2h}((x_{1}-x_{0})^{2}+(y_{1}-y_{0})^{2}+(\varphi_{1}-\varphi_{0})^{2})+h\cdot \frac{x_{1}+x_{0}}{4}
	\end{equation*}
	Moreover the discrete forces are given by
	\begin{equation*}
	F^+_d(q_{0},q_{1})=\frac{h}{2}\lambda \left( \mu_{x}d x_{1}+ \mu_{y} d y_{1} \right), \quad
	F^{-}_d(q_{0},q_{1})=\frac{h}{2}\lambda \left( \mu_{x}d x_{0}+ \mu_{y} d y_{0} \right),
	\end{equation*}
	with
	\begin{equation*}
		\begin{split}
			\lambda= & -\frac{\varphi_{1}-\varphi_{0}}{h^{2}}\left((x_{1}-x_{0})\cos\left(\frac{\varphi_{1}+\varphi_{0}}{2}\right)+(y_{1}-y_{0}) \sin\left(\frac{\varphi_{1}+\varphi_{0}}{2}\right)\right) \\
			& -\frac{1}{2}\sin\left(\frac{\varphi_{1}+\varphi_{0}}{2}\right)
		\end{split}
	\end{equation*}
	and
	\begin{equation*}
		\mu_{x}=\sin\left(\frac{\varphi_{1}+\varphi_{0}}{2}\right), \quad \mu_{y}=\cos\left(\frac{\varphi_{1}+\varphi_{0}}{2}\right).
	\end{equation*}
	With these ingredients we obtained an integrator with a nearly preservation of the energy (see Figure \ref{fig:test3}), where we use the Hamiltonian function
	\begin{equation*}
		H|_{\D^{*}}(x,\varphi,y,p_{1},p_{2})=\frac{1}{2}\left(\frac{p_{1}^2}{A(\varphi)}+p_{2}^2-x\right), \quad A(\varphi)=1+\frac{\sin^{2}(\varphi)}{\cos^{2}(\varphi)}.
	\end{equation*}
	\begin{figure}[htb!]
		\centering
		\includegraphics[width=0.7\linewidth]{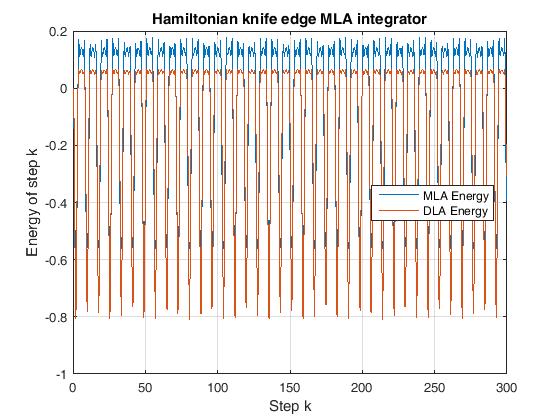}
		\caption{Experiment with the knife edge example: the initial positions are the origin $q_{0}=0$ and $q_{1}=(0.4,0.4,y_{1})$, the step is $h=0.5$ and the total number of steps is $N=600$.}
		\label{fig:test3}
	\end{figure}
\end{example}

\begin{example}
	We now slightly perturb the knife edge system by introducing the nonholonomic constraint (see \cite{modin})
	\begin{equation*}
	\sin(\varphi)\dot{x}-(\cos(\varphi)-\varepsilon)\dot{y}=0, \quad \varepsilon>0.
	\end{equation*}
	We obtain an integrator for the perturbed system that no longer preserves energy. Anyway, it still behaves clearly better than standard DLA algorithm (check Figure \ref{fig:test4}), for the Hamiltonian function
	\begin{equation*}
	H|_{\D^{*}}(x,\varphi,y,p_{1},p_{2})=\frac{1}{2}\left(\frac{p_{1}^2}{A(\varphi,\varepsilon)}+p_{2}^2-x\right), \quad A(\varphi,\varepsilon)=1+\frac{\sin^{2}(\varphi)}{(\cos(\varphi)-\varepsilon)^2}.
	\end{equation*}
	\begin{figure}[htb!]
		\centering
		\includegraphics[width=0.7\linewidth]{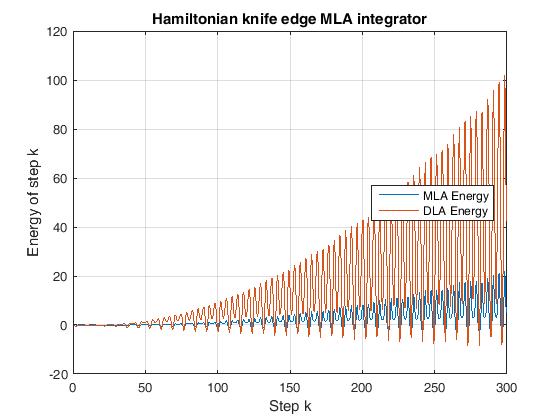}
		\caption{Experiment with the perturbed knife edge example with $\varepsilon=0.1$: the initial positions are the origin $q_{0}=0$ and $q_{1}=(0.4,0.4,y_{1})$, the step is $h=0.5$ and the total number of steps is $N=600$.}
		\label{fig:test4}
	\end{figure}
\end{example}

\section{Towards an intrinsic exact discrete flow}\label{towards}

We have recently introduced a formulation of nonholonomic mechanics using a suitable geometric environment, in this case, the skew-symmetric algebroid (cf. \cite{MR2492630,MR2660714}) which is a weaker version of the well-known concept of a Lie algebroid, where now the Lie bracket may not satisfy the Jacobi identity.

Following the program initiated by Alan Weinstein in \cite{MR1365779}, it was shown in \cite{MMdDM2006},\cite{MR3380059} and \cite{MMdDM2016} how to formulate  discrete mechanics in a unified way using the notion of a Lie groupoid. In the future, we want to find and study the equivalent algebraic structures for nonholonomic mechanics.

In this section we will describe some of the ingredients needed to develop this new theory, in particular, the nonholonomic exact discrete Lagrangian defined in $\M_{h}^{e,nh}$ and its main properties. 

Assume that we have a nonholonomic system defined by the triple $(Q, L, \D)$, where $L: TQ\rightarrow \R$ is a regular Lagrangian and $(L,\D)$ is a regular non-holonomic system. 

With the help of the \textit{constrained exact retraction}, defined by $R^{e-}_{h,nh}:\M_{h}^{e,nh}\rightarrow \U_h\subseteq \D$ introduced in Section \ref{sec:exponential-map}, we define  the \textit{nonholonomic exact discrete Lagrangian} for $(Q, L, \D)$ as a function on the exact discrete space $l_{d,nh}^{e,h}:\M_{h}^{e,nh}\rightarrow\R$ given by 
\begin{equation}
l_{d,nh}^{e,h}(q_{0},q_{1})=\int_{0}^{h} \left( L\circ \phi_{t}^{\Gamma_{nh}} \circ R^{e-}_{h,nh} \right) (q_0,q_1) \ dt.
\end{equation}
where $\{\phi_{t}^{\Gamma_{nh}}\}$ is the flow of $\Gamma_{nh}$, the solution of the nonholonomic dynamics.

To ease the notation let us introduce the following objects:
\begin{enumerate}
	\item given $(q_0, q_1)\in \M_{h}^{e,nh}$, define the following curves on $\D$ and $Q$, respectively: $$\gamma_{0}(t):=\left(\phi_{t}^{\Gamma_{nh}} \circ R^{e-}_{h,nh} \right) (q_0,q_1) \ \text{and} \ c_{0}(t):=\tau_{Q}\circ \gamma_{0}(t);$$
	\item a variation of the former curve is denoted by $$\gamma_{s}(t)=\left( \phi_{t}^{\Gamma_{nh}} \circ R^{e-}_{h,nh} \right) (q_0(s),q_1(s))\ \text{and} \ c_{s}(t):=\tau_{Q}\circ \gamma_{s}(t)$$
	\item the infinitesimal variation vector field on  the configuration manifold is  $$X_{0,1}(t)=\left. \frac{d}{ds} \right|_{s=0} c_{s}(t).$$
\end{enumerate}

Next we will prove a result which we will use later. The proof of this result involves the canonical involution $\kappa_{Q}:TTQ\rightarrow TTQ$ of the double tangent bundle. We recall that $\kappa_{Q}$ is a vector bundle isomorphism between the vector bundles $T\tau_{Q}:TTQ\rightarrow TQ$ and $\tau_{TQ}:TTQ\rightarrow TQ$. In fact, $\kappa_{Q}$ is characterized by the following condition if
\begin{equation*}
x:U\subseteq \R^{2} \rightarrow Q, \quad (s,t)\mapsto x(s,t)
\end{equation*}
is a smooth map then
\begin{equation*}
\kappa_{Q}\left( \frac{d}{dt}\frac{d}{ds} x(s,t) \right)=\frac{d}{ds}\frac{d}{dt} x(s,t).
\end{equation*}
So, $\kappa_{Q}^{2}=Id$. Moreover, if $X:Q\rightarrow TQ$ is a vector field on $Q$ then the tangent map $TX:TQ\rightarrow TTQ$ is a section of the vector bundle $T\tau_{Q}:TTQ\rightarrow TQ$ and, in addition, $\kappa_{Q}\circ TX=X^{C}$, where $X^{C}$ is the complete lift of $X$ to $TQ$ (see \cite{Tulcz} for more details).

\begin{lemma}\label{infinitesimal:lemma}
	Given a SODE $\Gamma$, if $\gamma_{s}$ is a one-parameter family of integral curves of $\Gamma$, then the infinitesimal variation vector field of $\gamma_{s}$ is the complete lift of the infinitesimal variation vector field of the one-parameter family of curves formed by the base integral curves of $\Gamma$, that is $c_{s}=\tau_{Q}\circ \gamma_{s}$.
\end{lemma}

\begin{proof}
	If $\gamma_{s}$ is a one-parameter family of integral curves of $\Gamma$, it has the form $\gamma_{s}=\frac{d}{dt}c_{s}$. Let $$X_{01}(t)=\left.\frac{d}{ds}\right|_{s=0}\tau_{Q}\circ\gamma_{s}(t)$$ be the infinitesimal variation vector field of $c_{s}$. Then the infinitesimal variation vector field of $\gamma_{s}$ is 
	\begin{eqnarray*}
		\left.\frac{d}{ds}\right|_{s=0}\gamma_{s}(t)&=&\left.\frac{d}{ds}\right|_{s=0}\frac{d c_{s}}{dt}(t)\\
		&=&\kappa_{Q} \left( \frac{d}{dt} \left.\frac{d}{ds}\right|_{s=0}c(s,t) \right)=\kappa_{Q} \left( \frac{d}{dt}X_{01}(t) \right)=X_{01}^{C}(t).
	\end{eqnarray*}
\end{proof}

Next, we will obtain an interesting expression for the differential of the nonholonomic exact discrete Lagrangian function $l_{d,nh}^{e,h}$. For this purpose, we will denote by $F_{nh}:\D\rightarrow T^{*}Q$ the continuous-time nonholonomic external force (see Remark \ref{rem1}).

\begin{proposition}\label{dis-lag}
	The differential of the nonholonomic exact discrete Lagrangian satisfies
	\begin{eqnarray*}
		\langle d l_{d,nh}^{e,h}(q_{0},q_{1}),(X_{q_0},X_{q_1}) \rangle &=&-\langle \beta_{nh}(q_{0},q_{1}),(X_{q_0},X_{q_1}) \rangle\\
		&&\hspace{-1cm}+\langle \F L\circ R_{h,nh}^{e+}(q_0, q_1),X_{q_1} \rangle- \langle\F L\circ R_{h,nh}^{e-}(q_0,q_1),X_{q_0} \rangle,
	\end{eqnarray*}
	where 
	$$\langle \beta_{nh}(q_{0},q_{1}), (X_{q_0}, X_{q_1})\rangle=\int_{0}^{h} \langle F_{nh}(\gamma_0(t)),X_{01}(t)\rangle \ dt$$ 	
	and we are identifying the vector $(X_{q_0},X_{q_1})\in T_{(q_0, q_1)}\M_{h}^{e,nh}$ with its image by $Ti:T\M_{h}^{e,nh}\hookrightarrow T(Q\times Q)$, with $i:\M_{h}^{e,nh}\hookrightarrow Q \times Q$ the canonical inclusion. The smooth curve $X_{01}: [0,h]\rightarrow TQ$ is defined as
	\[
	X_{01}(t)= T_{(q_0, q_1)}(\tau_{Q}\circ \phi_{t}^{\Gamma_{nh}} \circ R^{e-}_{h,nh}) (X_{q_0},X_{q_1})\; .
	\]
\end{proposition}

\begin{proof}
	Let $v: (-s,s)\rightarrow \M_{h}^{e,nh}$ be a smooth curve denoted by $v(s)=(q_0(s),q_1(s))$ such that $v(0)=(q_0,q_1)\in \M_{h}^{e,nh}$ and $v'(0)=(X_{q_0},X_{q_1})\in T_{(q_0, q_1)}\M_{h}^{e,nh}$ and 
	$$\gamma_{s}(t)=\left( \phi_{t}^{\Gamma_{nh}} \circ R^{e-}_{h,nh} \right) (q_0(s),q_1(s)).$$
	Then, using Lemma \ref{infinitesimal:lemma}, we have that
	\begin{equation}
	\begin{split}
	\langle d l_{d,nh}^{e,h}(q_0,q_1), & \left.\frac{d}{d s}\right|_{s=0}(q_{0}(s),q_{1}(s)) \rangle = \\
	& = \int_{0}^{h}\langle dL(\gamma_{0}(t)),\left.\frac{d}{d s}\right|_{s=0} \gamma_{s}(t) \rangle dt\\
	& = \int_{0}^{h}\langle dL(\gamma_{0}(t)),X_{01}^{C}(t) \rangle dt.
	\end{split}
	\end{equation}
	Note that $X_{01}^{C}(t)$ is a vector field on $TQ$ along  $\gamma_{0}(t)$, hence using \eqref{LdA:vector:field} it follows that  
	\begin{equation}
	\begin{split}
	&	\langle d l_{d,nh}^{e,h}(q_0,q_1), (X_{q_0}, X_{q_1})\rangle
	= X_{01}^{V}(h)(L)-X_{01}^{V}(0)(L)-\int_{0}^{h} \langle F_{nh}(\gamma_0(t)),X_{01}(t) \rangle dt\\
	& = \langle\F L(\gamma_{0}(h)),X_{01}(h)\rangle-\langle\F L(\gamma_{0}(0)),X_{01}(0)\rangle-\int_{0}^{h} \langle F_{nh}(\gamma_0(t)),X_{01}(t) \rangle dt.
	\end{split}
	\end{equation}
	By unyielding the definition of $X_{01}$ and identifying $(X_{q_0},X_{q_1})$ with its image by $Ti:T\M_{h}^{e,nh}\hookrightarrow T(Q\times Q)$, we see that
	\begin{equation*}
	\begin{split}
	& X_{01}(h)= T_{(q_0, q_1)}(\tau_{Q}\circ  R^{e+}_{h,nh}) (X_{q_0}, X_{q_1})=X_{q_1}, \\
	& X_{01}(0)=T_{(q_0, q_1)}( \tau_{Q} \circ R^{e-}_{h,nh}) (X_{q_0},X_{q_1})=X_{q_0},
	\end{split}
	\end{equation*}
	since
	\begin{equation*}
	\tau_{Q}\circ R^{e+}_{h,nh}=\text{pr}_{2}|_{\M_{h}^{e,nh}} \quad \text{and} \quad \tau_{Q}\circ R^{e-}_{h,nh}=\text{pr}_{1}|_{\M_{h}^{e,nh}},
	\end{equation*}
	where $\text{pr}_{1,2}:Q\times Q\rightarrow Q$ are the projection onto the first and second factor, respectively.
\end{proof}

Observe that in the previous Proposition, the intrinsic discrete objects associated to the nonholonomic problem are $d l_{d,nh}^{e,h}$, $\beta_{nh}\in \Lambda^1\M_{h}^{e,nh}$. Then, $\sigma_{nh}$ given by
\begin{equation}\label{sigmanh}
\sigma_{nh} (X_{q_0}, X_{q_1})=\langle (\F L\circ R_{h,nh}^{e+})(q_0, q_1),X_{q_1} \rangle- \langle(\F L\circ R_{h,nh}^{e-})(q_0,q_1),X_{q_0} \rangle
\end{equation}
is also a 1-form in $\M_{h}^{e,nh}$, where $(X_{q_0},X_{q_1})$ is identified with its image by $Ti$. From the definition of the Legendre transform $\F L: TQ\rightarrow T^*Q$, it is easy to see that this map can be extended to a map
\[
\widetilde{\sigma_{nh}}: \M_{h}^{e,nh}\longrightarrow T^*(Q\times Q)
\]
defined by expression \eqref{sigmanh} but applying it to an arbitrary vector $(X_{q_0}, X_{q_1})\in T_{(q_0,q_1)}(Q\times Q)$ with $(q_0, q_1)\in \M_{h}^{e,nh}$.

Finally we will relate the exact discrete objects we use in the modified Lagrange-d'Alembert principle in Section \ref{discrete:nonholonomic} with the intrinsic exact discrete objects defined above.

\begin{proposition}	
	The restriction to $\M_{h}^{e,nh}$ of the forced exact discrete Lagrangian function $L_{d,\widetilde{F_{nh}}}^{e,h}$ is just the non-holonomic exact discrete Lagrangian function $l_{d,nh}^{e,h}$, that is,
	\begin{equation*}
	\left. L_{d,\widetilde{F_{nh}}}^{e,h} \right|_{\M_{h}^{e,nh}}=l_{d,nh}^{e,h}.
	\end{equation*}
	
	Moreover, if $(q_{0},q_{1})\in \M_{h}^{e,nh}$ and $(X_{q_0}, X_{q_1})\in T_{(q_{0},q_{1})}\M_{h}^{e,nh}$ then $$\langle((\widetilde{F_{nh}})^{e,-}_d(q_0, q_1), (\widetilde{F_{nh}})^{e,+}_d(q_0, q_1)), (X_{q_0}, X_{q_1})\rangle=\langle \beta_{nh}(q_0, q_1),  (X_{q_0}, X_{q_1})\rangle.$$
\end{proposition}

\begin{proof}
	Given a pair of points $(q_0, q_1)\in \M_{h}^{e,nh}$, since the unique trajectory of $\Gamma_{nh}$ connecting the two points is also the unique trajectory of the forced problem $(L,\widetilde{F_{nh}})$ connecting these points, the expressions of $\left. L_{d,\widetilde{F_{nh}}}^{e,h} \right|_{\M_{h}^{e,nh}}$ and $l_{d,nh}^{e,h}$ match.
	
	According to Proposition \ref{dis-lag} and the observations following it we have that
	\begin{equation*}
	d l_{d,nh}^{e,h}+\beta_{nh}=\sigma_{nh}.
	\end{equation*}
	Then, since $\sigma_{nh}=i^{*}\widetilde{\sigma}_{nh}$ we have that
	\begin{equation*}
	i^{*}dL_{d,\widetilde{F_{nh}}}^{e,h}+\beta_{nh}=i^{*}\widetilde{\sigma}_{nh},
	\end{equation*}
	where $i:\M_{h}^{e,nh}\rightarrow Q\times Q$ is the inclusion. So,
	\begin{equation*}
	\beta_{nh}=i^{*}(\widetilde{\sigma}_{nh}-dL_{d,\widetilde{F_{nh}}}^{e,h}).
	\end{equation*}
	Observe that
	\begin{equation*}
	\widetilde{\sigma}_{nh}-dL_{d,\widetilde{F_{nh}}}^{e,h}=(-\F L \circ R_{h,nh}^{e,-}-D_{1}L_{d,\widetilde{F_{nh}}}^{e,h}, \F L \circ R_{h,nh}^{e,+}-D_{2}L_{d,\widetilde{F_{nh}}}^{e,h}).
	\end{equation*}
	Therefore, using Lemma \ref{exactlegendre}, we conclude
	\begin{equation*}
	\widetilde{\sigma}_{nh}-dL_{d,\widetilde{F_{nh}}}^{e,h}=((\widetilde{F_{nh}})^{e,-}_d, (\widetilde{F_{nh}})^{e,+}_d).
	\end{equation*}
\end{proof}
\section{Conclusion and future work}

In this paper, we have precisely identified the exact discrete equations for a nonholonomic system. The main ingredients were the definition of the exponential map for a constrained second-order differential equation allowing us to define the exact discrete  nonholonomic  constraint submanifold. Then, we define the main discrete elements that appear on the definition of the exact discrete nonholonomic equations. The special form of these equations allow us to introduce a new family of nonholonomic integrators showing in numerical computations the excellent behaviour of the energy.

In a future paper, we will study an  intrinsic version of discrete nonholomic mechanics in ${\M_{h}^{e,nh}}$ following the  steps given in Section \ref{towards}. Also we aim to find a nonholonomic version of Theorem \ref{variational-error} once we know  the exact discrete nonholonomic flow and the elements that it is necessary to approximate (discrete constraint submanifold, discrete Lagrangian and discrete forces) . Knowing these data we will be in a position describe  the order of the numerical method for a nonholonomic system as in the pure variational case. 
Moreover, since typically nonholonomic systems admit symmetries \cite{Bloch}, we will study the reduction of the discrete counterparts following the results by \cite{IMdDM2008}.

\section*{Acknowledgements}

D. Mart{\'\i}n de Diego and A. Simoes are supported  by I-Link Project (Ref: linkA20079), Ministerio de Ciencia e Innovaci\'on ( Spain) under grants  MTM\-2016-76702-P and ``Severo Ochoa Programme for Centres of Excellence'' in R\&D (SEV-2015-0554). A. Simoes is supported by the FCT research fellowship SFRH/BD/129882/2017. J.C. Marrero  acknowledges the partial support by European Union (Feder) grant PGC2018-098265-B-C32.

\bibliography{thesisreferences-1}{}

\end{document}